\documentclass[lettersize,journal]{IEEEtran}
\usepackage{amsmath, amsfonts, amssymb}
\usepackage{algorithmic}
\usepackage{array}
\usepackage{subcaption}
\usepackage{url}
\usepackage{verbatim}
\usepackage{graphicx}

\usepackage[english]{babel}
\usepackage{amsthm}
\usepackage{mathrsfs}
\usepackage{lastpage}
\usepackage{fancyhdr}
\usepackage{cite}
\usepackage{amsmath}
\usepackage[utf8]{inputenc}
\usepackage{mathtools}
\usepackage[ruled,linesnumbered]{algorithm2e}
\usepackage{textcomp}
\usepackage{svg}
\usepackage{ragged2e}
\usepackage{hyperref}
\usepackage{booktabs}
\usepackage{diagbox}
\def\BibTeX{{\rm B\kern-.05em{\sc i\kern-.025em b}\kern-.08em
    T\kern-.1667em\lower.7ex\hbox{E}\kern-.125emX}}
\usepackage{balance}

\newtheorem{theorem}{Theorem}

\newtheorem{definition}{Definition}[section]

\DeclarePairedDelimiter\ceil{\lceil}{\rceil}

\begin{document}

\newtheorem{observation}{Observation}[]
\newcommand{\justification}{\noindent
$\textbf{Justification}. \ $}
\theoremstyle{definition}
\newtheorem{Example}{Example}[]
\newtheorem{assumption}{Assumption}

\newcommand{\ith}[1]{${#1}^{\textit{th}} $}
\def \ie{i.e.,~}
\def \eg{e.g.,~}

\newcommand{\todo}{{\textbf{\color{red}[TO-DO]:}}}
\newcommand{\etal}{\textit{et al}.}

\newcommand{\figref}[1]{figure \ref{fig:#1}}
\newcommand{\Figref}[1]{Figure \ref{fig:#1}}
\newcommand{\tabref}[1]{Table~\ref{tab:#1}}
\newcommand{\semicolonmath}{\ \textbf{;} \ }
\newcommand{\conditionX}{
\Big|_{\substack{\textbf{x} = \textbf{x}^{(k+1), NORTH}}}
}

\newcommand{\xkth}{\textbf{x}^{(k)} }
\newcommand{\Pkth}{\textbf{P}^{(k)} }

\pagestyle{headings}

\title{Joint Optimization of Continuous Variables and Priority Assignments for Real-Time Systems with Black-box Schedulability Constraints}

\author{
Sen Wang,
Dong Li,
Shao-Yu Huang,
Xuanliang Deng,
Ashrarul H. Sifat, \\
Changhee Jung,
Ryan Williams,
Haibo Zeng\thanks{
This work is partially supported by NSF Grants No. 1812963 and 1932074.
Sen Wang, Dong Li, Xuanliang Deng, Ryan Williams, and Haibo Zeng are with Virginia Tech, Blacksburg, USA;
Shao-Yu Huang and Changhee Jung are with Purdue University, West Lafayette, USA. contact emails: swang666@vt.edu, hbzeng@vt.edu.
} 
}

\maketitle

\begin{abstract}
In real-time systems optimization, designers often face a challenging problem posed by the non-convex and non-continuous schedulability conditions, which may even lack an analytical form to understand their properties. 
To tackle this challenging problem, we treat the schedulability analysis as a black box that only returns true/false results. We propose a general and scalable framework to optimize real-time systems, named Numerical Optimizer with Real-Time Highlight (NORTH). 
NORTH is built upon the gradient-based active-set methods from the numerical optimization literature but with new methods to manage active constraints for the non-differentiable schedulability constraints.
In addition, we also generalize NORTH to NORTH+, to collaboratively optimize certain types of discrete variables (\eg priority assignments, categorical variables) with continuous variables based on numerical optimization algorithms.
We demonstrate the algorithm performance with two example applications: energy minimization based on dynamic voltage and frequency scaling (DVFS), and optimization of control system performance. In these experiments, NORTH achieved $10^2$ to $10^5$ times speed improvements over state-of-the-art methods while maintaining similar or better solution quality. NORTH+ outperforms NORTH by 30\% with similar algorithm scalability. Both NORTH and NORTH+ support black-box schedulability analysis, ensuring broad applicability.
\end{abstract}


\section{Introduction}

In recent years, the real-time systems community has developed many impressive scheduling algorithms and schedulability analysis techniques. 
However, performing optimization with these schedulability constraints is challenging because these constraints are often non-differentiable and/or non-convex, such as the ceiling function for response time calculation~\cite{audsley1993applying}. 
Some other schedulability analyses are even more challenging because they do not have an analytical form, such as those based on demand-bound functions~\cite{Baruah1990PreemptivelySH}, real-time calculus~\cite{Thiele2000}, abstract event model interfaces~\cite{henia2005system}, and timed automata~\cite{Larsen1997UppaalIA, Nasri2019ResponseTimeAO}. 

Optimizing real-time systems faces the additional challenge of their ever-increasing complexity.
The functionality of modern real-time systems is rapidly expanding~\cite{Heintzman2021-mw}, and there may be hundreds of software tasks~\cite{Kramer15benchmark} and even more runnable in automotive systems. 
The underlying hardware and software systems are also becoming more sophisticated, such as heterogeneous computing platforms, specialized hardware accelerators, and domain-specific operating systems. 

The existing real-time system optimization algorithms cannot adequately address the above two challenges as they lack scalability and/or applicability. 
For example, it is often challenging to model complicated schedulability constraints in standard mathematical optimization frameworks. 
Although some customized optimization frameworks \cite{zhao2018unified, Zhao2020AnOF, zhao2022design} have been designed for real-time systems, these frameworks usually rely on special assumptions, such as \emph{sustainable} schedulability analysis\footnote{Sustainable analysis means that under this analysis, a schedulable task set should remain schedulable when its parameters become ``better'', \eg shorter WCET or longer period~\cite{baruah2006sustainable, Burns2008SustainabilityIR}.} with respect to the design variables (\eg periods~\cite{zhao2018unified} and worst-case execution times (WCETs)~\cite{Zhao2020AnOF, zhao2022design}).
Such special assumptions may not hold in certain schedulability analysis~\cite{Nasri2019ResponseTimeAO} or could be difficult to verify in general schedulability analysis.

This paper proposes a new optimization framework for real-time systems with \textit{blackbox} schedulability analysis that only returns true/false results. 
The framework, called Numerical Optimizer with Real-Time Highlight (NORTH), utilizes numerical optimization methods for real-time systems optimization. 
Compared to alternative optimization frameworks such as Integer Linear Programming (ILP), numerical optimization methods are more general (they work with objective functions or constraints of many forms) and have good scalability for large-scale optimization problems. 
Unfortunately, existing numerical methods cannot be directly applied to optimize real-time systems. 
For example, gradient-based methods often rely on well-defined gradient information, but many schedulability constraints are not differentiable. 
Targeted at the issues above, we modify the gradient-based method to avoid evaluating the gradient of schedulability constraints and then propose a novel technique, variable elimination (VE), to enable searching for better solutions along the boundary of schedulable solution space. These techniques are further generalized to assist in optimizing both continuous variables and priority assignments, a common type of discrete variables in real-time systems.


To the best of our knowledge, this paper presents the first optimization framework for real-time systems based on numerical algorithms. The framework is applied to two example problems: energy minimization based on DVFS and control system performance optimization. Compared with the alternatives, NORTH and NORTH+ have the following advantages:
\begin{itemize}
\item \textit{Generality}: NORTH supports optimization with any schedulability analysis that provides true/false results. Furthermore, NORTH+ also supports optimizing a mix of continuous variables and priority assignments.
    
\item \textit{Scalability}: NORTH and NORTH+ have good scalability because the frequency of schedulability analysis calls typically scales polynomically with the number of variables.
    
\item \textit{Quality}: NORTH and NORTH+ adapt classical numerical optimization methods and propose new approaches to improve performance further. 
It is shown to outperform gradient-based optimizers and state-of-the-art methods.
\end{itemize}

\textbf{Extensions against previous work~\cite{Wang2023RTAS}.} 
This journal paper extends its previous conference paper~\cite{Wang2023RTAS} to simultaneously optimize continuous variables and priority assignments (Section~\ref{section_hybrid_optimization}). The framework NORTH in the conference paper~\cite{Wang2023RTAS} only handles continuous variables. To address this fundamental problem in fixed-priority scheduling, we propose an iterative priority assignment algorithm guided by numerical optimization, which is subsequently integrated into the NORTH framework to create NORTH+. Experimental results (Section~\ref{exp_north_control}) indicate that NORTH+ delivers a 30\% performance improvement over NORTH by additionally optimizing priority assignments, while maintaining a comparable scalability to NORTH. Furthermore, this paper provides more discussion on the generalization, limitations, and potential enhancements of both NORTH and NORTH+ (Section~\ref{applications}).

\section{Related Work}
The topic of optimization for real-time systems has been studied in many papers. 
Broadly speaking, they can be classified into four categories~\cite{Zhao2020AnOF}: (1) meta-heuristics~\cite{Shin2007OPTIMALPA, Tindell2004AllocatingHR} such as simulated annealing; 
(2) direct usage of standard mathematical optimization frameworks such as branch-and-bound (BnB)~\cite{Jonsson1997APB}, ILP~\cite{zeng2012efficient}, and convex programming~\cite{Aydin2006SystemLevelEM}; 
(3) problem-specific (\eg minimizing energy in systems with DVFS) methods~\cite{Bambagini2016EnergyAwareSF}; 
(4) customized optimization frameworks including~\cite{zhao2017virtual, Zhao2020AnOF, zhao2022design}. 
The meta-heuristic methods are relatively slow and are often outperformed by other methods. 
The application of other methods is limited because they usually rely on special properties of the schedulability analysis. 
For example, ILP requires that schedulability analysis can be transformed into linear functions, while \cite{Zhao2020AnOF, zhao2022design} depends on sustainable schedulability analysis.

Numerical optimization methods are widely applied in many situations because of their generality and scalability. 
Classical methods include active-set methods (ASM) and interior point methods (IPM)~\cite{Nocedal2006NumericalO2}, and their recent extensions, such as Dog-leg~\cite{POWELL197031}, Levenberg-Marquardt~\cite{Marquardt1963AnAF}, and gradient-projection~\cite{Balashov2019GradientPA}. 
However, these optimization methods cannot be directly utilized for non-differentiable schedulability constraints. 
Although numerical gradient could be helpful, it may become misleading at non-differentiable points. 
Gradient-free methods, \eg model-based interpolation and the Nelder–Mead method~\cite{Nelder1965ASM}, usually run much slower and are less well-studied than gradient-based methods~\cite{Nocedal2006NumericalO2}. 
Some recent studies utilize machine learning to solve some real-time system problems~\cite{Lee2021MLFR, Bo2021DevelopingRS}. However, these methods may make the design process more time-consuming due to the challenge of preparing a large-scale training dataset.

Energy minimization based on Dynamic Voltage and Frequency Scaling (DVFS) has been extensively studied for systems with different scheduling algorithms~\cite{Yao1995ASM, Pillai2001RealtimeDV, Aydin2006SystemLevelEM, Lee2004OnlineDV, Bini2009MinimizingCE, li2024energyefficientcomputationdvfsusing} and more sophisticated system models, such as directed acyclic graph (DAG) models on multi-core~\cite{Bhuiyan2020EnergyEfficientPR}, mixed-criticality scheduling~\cite{Bhuiyan2020OptimizingEI}, and limited-preemptive DAG task models~\cite{Nasri2019ResponseTimeAO}. 
These schedulability analysis methods could be complicated and lack the special properties upon which state-of-art approaches rely. For instance, the schedulability analysis in Nasri \etal~\cite{Nasri2019ResponseTimeAO} is not sustainable. 

The second example application is the optimization of control performance~\cite{Mancuso2014OptimalPA}, which is usually modeled as a function of periods and response times~\cite{Zhao2020AnOF}. Various algorithms have been proposed for this problem, such as mixed-integer geometric programming~\cite{Davare2007PeriodOF}, genetic algorithm~\cite{Shin2007OPTIMALPA}, BnB~\cite{Bini2005OptimalTR, Mancuso2014OptimalPA}, and some customized optimization frameworks~\cite{zhao2018unified, Zhao2020AnOF}. 
However, the applications of these works are limited to certain schedulability analyses or system models.

\section{System model}
This section introduces common notations and the problem description for continuous variable optimization. The system model for co-optimizing continuous variables and priority assignments will be introduced later.
\subsection{Notations}
In this paper, scalars are denoted with light symbols, while vectors and matrices are denoted in bold. Subscripts such as $i$ usually represent the \ith{i} element within vectors.
We denote the iteration number in a superscript in parentheses during optimization iterations. For example, $\textbf{x}_i^{(0)}$ denotes the \ith{i} element of a vector $\textbf{x}$ at the \ith{0} iteration.
We use $\|\textbf{v} \|$ to denote the Euclidean norm of a vector $\textbf{v}$, $\|\textbf{v} \|_1$ for norm-1, $|\textbf{S}|$ for the number of elements in a set $\textbf{S}$, and $|x|$ for the absolute value.
We typically use $h$ to denote numerical granularity ($10^{-5}$ in experiments), $\textbf{x}$ for variables, $\mathcal{F}(\textbf{x})$ for objective functions, and $\boldsymbol{\nabla} \mathcal{F}(\textbf{x})$ for gradient.
In task sets, we usually use $\tau_i$ to denote a task and $r_i(\textbf{x})$ for a response time function of $\tau_i$ that depends on the variable $\textbf{x}$.
In the context of applications, we adhere to the standard notation within the specific domain to avoid any potential confusion.

In cases where $\mathcal{F}(\textbf{x})$ has a sum-of-square form:
\begin{equation}
  \min_{\textbf{x}} \sum_i \mathcal{F}^2_i(\textbf{x})
  \label{sum_of_item}
\end{equation}
then we have a Jacobian matrix $\textbf{J}$ defined as follows:
\begin{equation}
    \textbf{J}_{ij} = \frac{\partial \mathcal{F}_i(\textbf{x})}{\partial \textbf{x}_j}
    \label{jacobian_formal}
\end{equation}
where $\textbf{J}_{ij}$ is the entry of $\textbf{J}$ at the \ith{i} row and the \ith{j} column.
$\textbf{J}$'s transpose is denoted as $\textbf{J}^T$.


\subsection{Concepts from Numerical Optimization}
\label{section_terminology}
For clarity, we explain some terms in numerical optimization following Nocedal~\etal~\cite{Nocedal2006NumericalO2}:
\begin{definition} [Differentiable point]
     If the objective function $\mathcal{F}(\textbf{x})$ is differentiable at $\textbf{x}$, then $\textbf{x}$ is a differentiable point.
\end{definition}

\begin{definition} [Descent vector]
    A vector $\boldsymbol{\Delta}$ is called a descent vector for function $\mathcal{F}(\textbf{x})$ at $\textbf{x}$ if 
    \begin{equation}
        \mathcal{F}(\textbf{x} + \boldsymbol{\Delta}) < \mathcal{F}(\textbf{x})
    \end{equation}
\end{definition}

\begin{definition}[Descent direction]
A vector $\boldsymbol{\Delta}$ is a descent direction if there is $\alpha>0$ such that $\alpha \boldsymbol{\Delta}$ is a descent vector.    
\end{definition}


\noindent\textbf{Active/inactive constraints}:
An inequality constraint $g(\textbf{x}) \leq 0$ is \textit{active} at a point $\textbf{x}^*$ if $g(\textbf{x}^*) = 0$. 
$g(\textbf{x}) \leq 0$ is \textit{inactive} if it holds with strict ``larger/smaller than'' (\ie $g(\textbf{x}^*) < 0$ at $\textbf{x}^*$).
Equality constraints are always active constraints.

\noindent\textbf{Active-set methods (ASM)}: 
In constrained optimization problems, ASM only considers active constraints when finding an update step in each iteration.
This is based on the observation that inactive constraints do not affect an update if the update is small enough. Examples of ASM include the simplex method and the sequential-quadratic programming.

\noindent\textbf{Trust-region methods (TRM)}: TRM are optimization techniques that iteratively update a searching region around the current solution. 
TRM adapts the size of the trust region during each iteration, balancing between exploring new solutions within the trust region and ensuring that the objective function within the trust region is accurately approximated. 


\subsection{Problem Formulation}
We consider a real-time system design problem as follows:
\begin{align}
     \min_{\textbf{x}} \ \ \ \ & \ \mathcal{F}=\sum_i \mathcal{F}_i(\textbf{x})
    \label{general_F} \\
   \textit{subject to}\ \ & \   \text{Sched}(\textbf{x}) = 0
\label{schedulability_analysis_true_false} \\
   & \ \text{lb}_i \leq \textbf{x}_i \leq \text{ub}_i
    \label{general_inequality_constraint}
\end{align}
\noindent where $\textbf{x} \in \mathbb{R}^N$ denotes the optimization variables or design choices, such as run-time frequency, task periods, or priority assignments; $N$ denotes the number of decision variables; $\text{lb}_i$ and $\text{ub}_i$ denote $\textbf{x}_i$'s lower bound and upper bound, respectively; $\mathcal{F}_i(\textbf{x}): \mathbb{R}^N \xrightarrow[]{} \textbf{R}$ denotes the objective function, such as energy consumption. 
The schedulability analysis constraint is only required to return binary results:
\begin{equation}
    \text{Sched}(\textbf{x})=\begin{cases}
        0, & \text{system is schedulable}\\
        1, & \text{otherwise}
    \end{cases}
    \label{general_sched_function}
\end{equation}
The problem formulation above is generally applicable in many challenging situations where the schedulability analysis does not have analytical forms, such as demand bound functions~\cite{Baruah1990PreemptivelySH} or timed automata~\cite{Larsen1997UppaalIA, Nasri2019ResponseTimeAO}.

\begin{assumption}
\label{assumption_feasible}
A feasible initial solution is available. 
\end{assumption}
For instance, in DVFS, the maximum CPU frequency could be the feasible initial solution in many situations. Otherwise, the problem may not have feasible solutions that can be easily found.

\begin{assumption}
\label{assumption_continuous}
     The variables $\textbf{x}$ are continuous variables, such as task WCETs and periods. 
\end{assumption}
Note that Assumption~\ref{assumption_continuous} does \textit{not} assume the constraints or objective functions to be differentiable with respect to $\textbf{x}$.

Assumption~\ref{assumption_continuous} only applies to the optimization framework NORTH introduced in Section~\ref{section_nmbo} and~\ref{section_how_to_eliminate} but not for NORTH+ in Section~\ref{section_hybrid_optimization}, which introduces how to optimize priority assignments.
Section~\ref{applications} discusses how to relax these two assumptions. 

\noindent \textbf{Challenges.} The primary challenge is the blackbox schedulability constraint \eqref{schedulability_analysis_true_false} which cannot provide gradient information.
However, most numerical optimization algorithms are proposed for problems with differentiable objective functions and/or constraints. Straightforward applications with numerical gradients may cause significant performance loss. 
Moreover, schedulability analysis can be computationally expensive in many cases. Therefore, minimizing the number of schedulability analysis calls is essential to achieve good algorithm scalability.

\subsection{Application: Energy Optimization}
\label{energy_opt_intro}
We first consider an energy minimization problem based on DVFS~\cite{Bambagini2016EnergyAwareSF}. 
DVFS reduces power consumption by running tasks at lower CPU run-time frequencies, albeit at the expense of longer response times. Therefore, it is important to guarantee the system's schedulability while optimizing energy consumption. 
Given a task set of $N$ periodic tasks, we hope to change the run-time frequency $\textbf{f}$ to minimize the energy consumption $E_i(\textbf{f})$ of each task $\tau_i$:
\begin{align}
 \min_{\textbf{f}}  \ \ \ \ & \  \sum_{i=0}^{N-1} (\sqrt{E_i(\textbf{f})})^2 
    \label{ls_energy}\\
    \textit{subject to}\ \ & \   \text{Sched}(\textbf{x}) = 0\\
    & \ \text{lb}_i \leq \textbf{f}_i \leq \text{ub}_i, i \in [0, N-1]
\end{align}
where the energy function $E_i(\textbf{f})$ is estimated as:
\begin{equation}
    E_i(\textbf{f}) =\frac{H}{T_i} (\beta_E + \alpha_E \textbf{f}_i^{\gamma_E}) \times c_i 
\end{equation}
where $H$ is the hyper-period (the least common multiple of tasks' periods), $T_i$ is $\tau_i$'s period. 
The energy function~\cite{Huang2014EnergyED, Guo2019EnergyEfficientRS} considers the static and dynamic power consumption~\cite{Guo2019EnergyEfficientRS, Pagani2013EnergyET} with parameters $\alpha_E = 1.76 \text{Watts} / \text{GHz}^3$, $\gamma_E=3$, $\beta_E=0.5$. 
The execution time $c_i$ of $\tau_i$ is determined by a frequency model~\cite{Bambagini2016EnergyAwareSF, Aydin2006SystemLevelEM}:
\begin{equation}
    c_i = c_i^{\text{fix}} + \frac{c_i^{\text{var}}}{\textbf{f}_i}
    \label{frequency_model}
\end{equation}
where both speed-independent and speed-dependent operations are considered.

The schedulability analysis can be provided by any method. For the sake of comparing with baseline methods, the response time analysis (RTA) model for fixed task priority scheduling in single-core and preemptive platforms~\cite{Joseph1986FindingRT} is first considered:
\begin{equation}
    r_i = c_i +\sum_{j \in \text{hp}(i)} \ceil{\frac{r_i}{T_j}}{c_j}
    \label{rta_LL}
\end{equation}
where $\text{hp}(i)$ denotes the tasks with higher priority than $\tau_i$. 
Denote $\tau_i$'s deadline as $D_i$, and we have
\begin{equation}
        \text{Sched}(\textbf{f})=\begin{cases}
        0, & \forall i, r_i(\textbf{f}) \leq D_i \\
        1, & \text{otherwise}
    \end{cases}
    \label{sched_model_ll}
\end{equation}

Another schedulability analysis considered is the model verification methods proposed by Nasri~\etal~\cite{Nasri2019ResponseTimeAO} for node-level preemptive DAG tasks.
These methods often provide less pessimism than many analytic approaches in real systems. However, many available methods~\cite{Zhao2020AnOF} cannot be applied to optimize with it due to the lack of analytical expressions and sustainability property.

\subsection{Application: Control Quality Optimization}
\label{control_opt_intro}
The second application focuses on optimizing the control performance~\cite{Zhao2020AnOF}. 
Similar to the problem description in Zhao~\etal~\cite{Zhao2020AnOF}, the control system performance is approximated by a function of period $\textbf{T}_i$ and response time $r_i(\textbf{T}, \textbf{P})$:
\begin{align}
    \min_{\textbf{T}, \textbf{P}} \ \ \ \ & \  \sum_{i=0}^{N-1} (\sqrt{\alpha_i \textbf{T}_i + \beta_i r_i(\textbf{T}, \textbf{P})+ \gamma_i r_i^2(\textbf{T}, \textbf{P}) } )^2 
    \label{ls_control}\\
     \textit{subject to}\ \ & \  \text{Sched}(\textbf{T}, \textbf{P}) = 0 \\
    & \ \text{lb}_i \leq \textbf{T}_i \leq \text{ub}_i, i \in [0, N-1]\label{ls_bound}
\end{align}
where $\textbf{T}_i$ is the task $\tau_i$'s period; $\textbf{P}$ denotes the task set's priority assignments; $\alpha_i$, $\beta_i$ and $\gamma_i$ are control system's weight parameters and can be estimated from experiments; $r_i(\textbf{T}, \textbf{P})$ is $\tau_i$'s response time. 
We used the schedulability analysis proposed in \cite{Nasri2019ResponseTimeAO} in our experiments, although any other forms of schedulability analysis are also supported.

\begin{figure}[ht]
\centering
\includegraphics[width=0.45\textwidth]{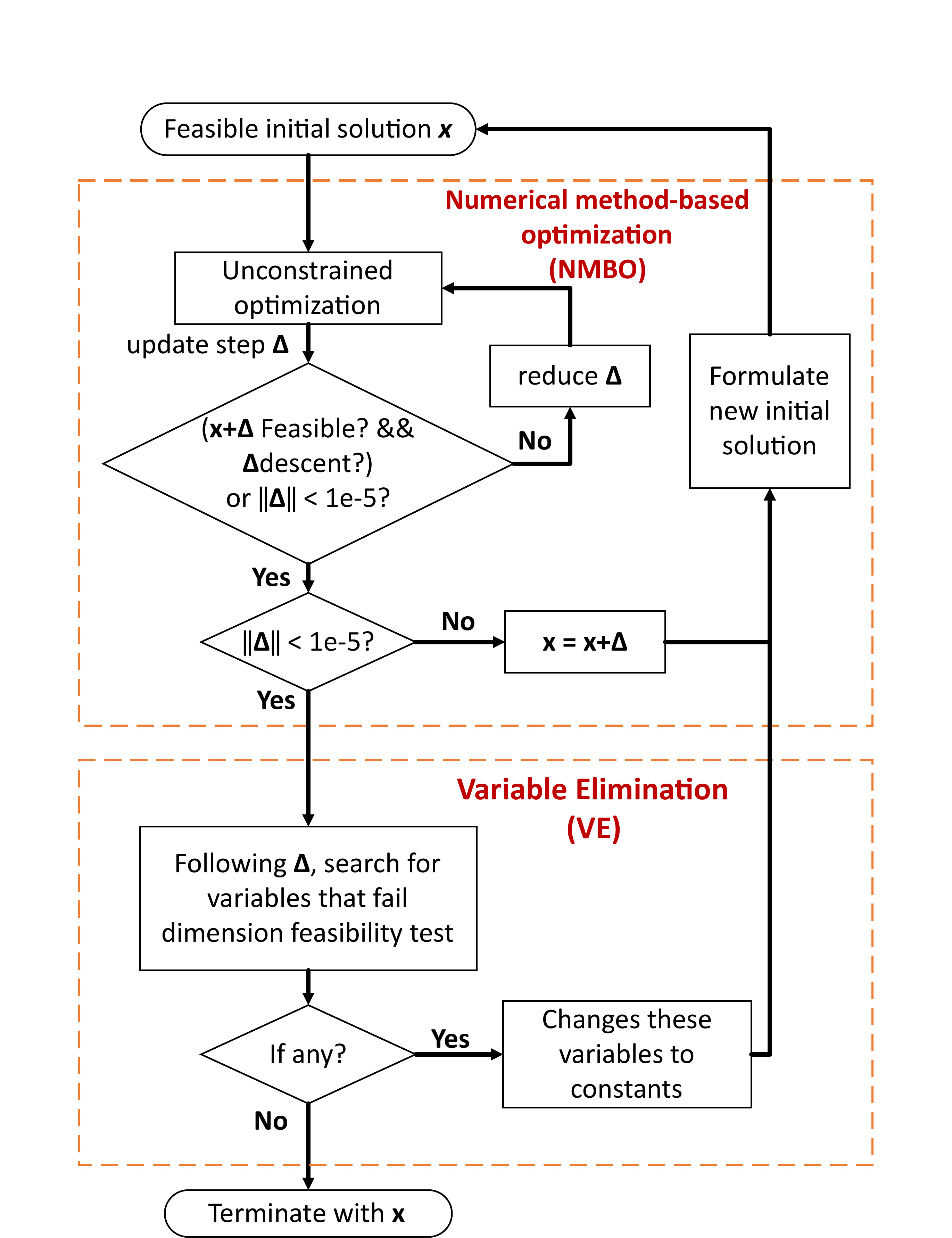}
\caption{NORTH framework and components overview.\\
\textbf{NMBO}: Beginning with a feasible solution $\textbf{x}$, NORTH first utilizes trust-region methods to find an update direction $\boldsymbol{\Delta}$ for
the optimization problem without the schedulability constraints. 
If $\boldsymbol{\Delta}$ leads $\textbf{x}$ into an infeasible region, then $\boldsymbol{\Delta}$ is kept decreased until $\textbf{x}+\boldsymbol{\Delta}$ is feasible. When $\boldsymbol{\Delta}$ becomes small enough such as $\| \boldsymbol{\Delta}\| < 10^{-5}$, performing further iterations only brings small performance improvements, and therefore we terminate the iterations and move on to the next step. \\
\textbf{VE}: After NMBO terminates, we check if there are close active constraints. If so, the involved variables are transformed into constants in future iterations. The algorithm terminates when there are no more variables to optimize.} 
\label{fig:main_framework_fig}
\end{figure}

\section{Numerical method-Based Optimization}
\label{section_nmbo}

This section introduces a Numerical Method-Based Optimization (NMBO) algorithm to address the challenge of non-differentiable schedulability constraints. 
Leveraging active-set methods (ASM) and trust-region methods (TRM), NMBO allows temporarily ``ignoring'' inactive constraints during iterations, therefore avoiding analyzing gradients of the non-differentiable schedulability constraints. 

The flowchart of NMBO is shown in Figure~\ref{fig:main_framework_fig}.
Each iteration begins with a feasible solution $\textbf{x}^{(k)}$ and performs unconstrained optimization to obtain an update step $\boldsymbol{\Delta}$. 
It then verifies whether $\textbf{x}^{(k)} + \boldsymbol{\Delta}$ is feasible and improves $\mathcal{F}(\textbf{x}^{(k)})$. If so, $\boldsymbol{\Delta}$ is accepted; 
otherwise, we decrease $\boldsymbol{\Delta}$ based on trust-region algorithms~\cite{Nocedal2006NumericalO2}. 
Finally, NMBO will terminate at either a stationary point (\ie points with zero gradients) or feasible region boundaries.

The unconstrained optimizer we use is Levenberg-Marquardt algorithm~\cite{Levenberg1944AMF, Marquardt1963AnAF} (LM), one of the classical trust-region algorithms.
In each iteration, LM updates with the following formula:
\begin{equation}
    (\textbf{J}^{(k)T} \textbf{J}^{(k)} + \lambda \ 
    \text{diag}(\textbf{J}^{(k)T} \textbf{J}^{(k)}))
    \boldsymbol{\Delta} = -\textbf{J}^{(k)T} \boldsymbol{\mathcal{F}}(\textbf{x}^{(k)})
    \label{lm_update}
\end{equation}
where $\textbf{J}^{(k)}$ is the Jacobian matrix after linearizing the objective function~\eqref{sum_of_item}.
We can control the parameter $\lambda$ in Equation~\eqref{lm_update} to change the length of $\boldsymbol{\Delta}$.
For example, larger $\lambda$ implies smaller $\boldsymbol{\Delta}$. 

\subsection{Why NMBO alone is not enough: Challenges of Numerical Gradient}
Numerical gradients can be calculated as follows when the analytical gradient cannot be derived (\eg black-box objective functions):
\begin{equation}
    \frac{\partial \mathcal{F}_i(\textbf{x})}{\partial \textbf{x}_i} =  \frac{\mathcal{F}_i(\textbf{x}_1,..,\textbf{x}_i+h, .., \textbf{x}_N) - \mathcal{F}_i(\textbf{x}_1,..,\textbf{x}_i-h, .., \textbf{x}_N)}{2h}
    \label{numerical_jacobian}
\end{equation}
where $h=10^{-5}$ in our experiments, more consideration on how to choose $h$ can be found in Nocedal~\etal~\cite{Nocedal2006NumericalO2}. 
The numerical gradients should be used with caution because:
\begin{observation}
    Numerical gradient \eqref{numerical_jacobian} at non-differentiable points could be a non-descent vector or lead to a non-feasible solution.
    \label{unreliable_Jacobian_obs}
\end{observation}
\begin{justification}
This is because the exact gradient at non-differentiable points is not well-defined and, therefore, is sensitive to many factors, such as the choice of $h$.
\end{justification}
\begin{observation}
    Gradient-based constrained optimizers may terminate at infeasible solutions if the constraints are not differentiable.
\end{observation}
\begin{justification}
The numerical gradient of non-differentiable constraints is not reliable.
\end{justification}

NMBO overcomes the challenges above by constructing local approximation models exclusively for objective functions rather than constraints. This is because schedulability constraints frequently lack differentiability, while the objective functions may not. 
If the objective function contains non-differentiable components, we use $0$ for its numerical gradient instead.
After obtaining a descent vector $\boldsymbol{\Delta}$, NMBO accepts it only if $\textbf{x}^{(k)} + \boldsymbol{\Delta}$ is feasible. Therefore, NMBO guarantees to find feasible results while maintaining the speed advantages of gradient-based optimization algorithms.

\subsection{Termination Conditions for NMBO}
NMBO will terminate if the relative difference in the objective function
\begin{equation}
    \delta_{\text{rel}} =\frac{\mathcal{F}(\textbf{x}^{(k+1)}) - \mathcal{F}(\textbf{x}^{(k)})}
    {\mathcal{F}(\textbf{x}^{(k)})} 
\end{equation}
becomes very small, \eg $\delta_{\text{rel}} \leq 10^{-5}$. NMBO will also terminate if the number of iterations exceeds a certain number (e.g., $10^3$).

\begin{Example}
\label{example_nmbo}
    Let's consider a simplified energy minimization problem:
    \begin{align}
\min_{\textbf{c}_1,  \textbf{c}_2} \ & (8\textbf{c}_1^{-1})^2 + (\textbf{c}_2^{-1})^2 \\
     \textit{subject to: } \ &
     \text{Sched}(\textbf{c}_1, \textbf{c}_2) = 0 \\
    & 4 \leq \textbf{c}_1 \leq 10 \\
    &1 \leq \textbf{c}_2 \leq 10
\end{align}
where the WCET of each task are the variables.

The task set includes two tasks: task 1 and task 2. 
Task 1's initial execution time, period, and deadline are $\{4, 10, 6\}$, respectively;
Task 2's initial execution time, period, and deadline are $\{1, 40, 40\}$, respectively;
Task 1 has a higher priority than task 2.
The schedulability analysis is based on Equation~\eqref{rta_LL} and \eqref{sched_model_ll}.


Let's consider an initial solution $\textbf{c}^{(0)} = (4, 1)$. We use LM to perform unconstrained optimization, and use $10^3$ as the initial $\lambda$ in equation~\eqref{lm_update}, then we can perform one iteration as follows:
\begin{equation}
    \textbf{J}^{(0)}
    =\begin{bmatrix}
    -8/4^2  &0\\0& -1/1^3
\end{bmatrix}
    = \begin{bmatrix}
    -0.5  &0\\0& -1
\end{bmatrix}
\end{equation}
\begin{equation}
\textbf{c}^{(1)} = \textbf{c}^{(0)} + \boldsymbol{\Delta}  = \begin{bmatrix}
    4.004 &1.001
\end{bmatrix}
\end{equation}
Multilpe iterations will be performed until $\textbf{c}^{(k)}$ becomes close to violating the schedulability constraints, \eg $(5.999, 1.499)$. At this point, NMBO cannot make big progress further without violating the schedulability constraints and so will terminate.
\end{Example}

\begin{observation}
    The schedulability analysis process in an iterative algorithm such as NMBO can often be sped up if it can utilize a ``warm start''.
\end{observation}
\begin{justification}
    Iterative algorithms typically make progress incrementally, which is beneficial for improving runtime speed in various schedulability analyses.
    For instance, consider the classical response time analysis~\eqref{rta_LL}. 
    After the execution time $\textbf{c}_i$ of $\tau_i$ increases to $\textbf{c}_i+\boldsymbol{\Delta}_i$ where $\boldsymbol{\Delta}_i >0$, $r_i(\textbf{c}_i)$ could be a warm start to analyze $r_i(\textbf{c}_i)+\boldsymbol{\Delta}_i$ during the fixed-point iterations.
    Please check Davis~\etal~\cite{Davis2008EfficientES} to learn more about this topic.
\end{justification}

\section{Variable Elimination}
\label{variable_elimination_section}

\begin{figure}[t]
\centering
\includegraphics[width=0.35\textwidth]{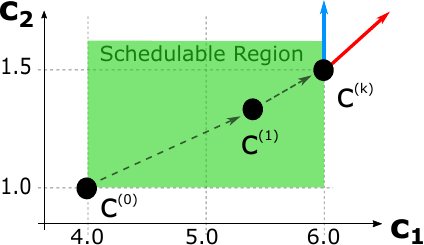}
\caption{Variable elimination motivating example.
Consider the problem in Example~\ref{example_nmbo} and assume NMBO terminates at $(5.999, 1.499)$.
The red arrow shows the update step $\boldsymbol{\Delta}$ from classical unconstrained optimizers such as equation~\eqref{lm_update}. Moving toward the red arrow will make task 1 miss its deadline.
However, we can improve $\textbf{c}^{(k)}$ without violating the schedulability constraints by only updating $\textbf{x}_2$ following $\boldsymbol{\Delta}$ while leaving $\textbf{x}_1$ unchanged (moving towards the blue arrow).
}
\label{fig:example_elimination}
\end{figure}

\subsection{Motivation and Definitions}
There are usually chances to improve performance further after NMBO terminates.
\Figref{example_elimination} shows an example from Example~\ref{example_nmbo}, where a feasible update direction (blue) is hidden in an infeasible descent vector $\boldsymbol{\Delta}$ (red).
Therefore, it is possible to optimize a subset of variables that will not immediately violate constraints, thereby enhancing the objective function. Such a subset cannot be directly obtained from NMBO because NMBO always updates all the variables simultaneously (for example, in equation~\eqref{lm_update}, all the variables are updated in each iteration). Therefore, we propose a new algorithm: variable elimination (VE).

To that end, detecting the feasible region's boundary during iterations is important.
Therefore, we generalize the definition of active constraints (originally defined for continuous functions) for the binary schedulability analysis:
\begin{definition}[Active schedulability constraint]
\label{def_active_sched_constraint}
A binary schedulability constraint, equation~\eqref{schedulability_analysis_true_false}, is an active constraint at a point $\textbf{x}^*$ if:
\begin{equation}
    \text{Sched}(\textbf{x}^*) =0
\end{equation}
\begin{equation}
    \exists \boldsymbol{\delta}, \|\boldsymbol{\delta} \| \leq h, \text{Sched}(\textbf{x}^*+\boldsymbol{\delta}) =1
     \label{larger_than_close}
\end{equation}
where $h>0$ is the numerical granularity.
\end{definition}
In other words, if a schedulability constraint is active at $\textbf{x}^*$, then $\textbf{x}^*$ is schedulable but is close to becoming unschedulable.

Upon termination of NMBO, some constraints may become active constraints and prevent NMBO from making progress.
One potential solution is removing these active constraints from the optimization problem. However, we must also guarantee that the final solutions respect these active constraints. Therefore, we can find the variables involved in the active constraints, lock their values from future optimization, and then safely ignore these active constraints in future iterations.
Based on Definition~\ref{def_active_sched_constraint}, we can find these variables based on the following method:
\begin{definition} [Dimension feasibility test]
\label{def_dimension_test}
    A solution $\textbf{x} \in \mathbb{R}^N$ for problem \eqref{general_F} passes dimension-j feasibility test of length $\textit{d}$ along the direction $\boldsymbol\Delta$ if 
    $\textbf{x} \oplus (\frac{ \boldsymbol\Delta_j}{|\boldsymbol\Delta_j|}d, j )$ is feasible, with $j \in [0, N-1]$.
\end{definition}
\noindent where the $\oplus$ operation for vectors $\textbf{x}$ and $\textbf{y}$ is defined as follows:
\begin{equation}
    \textbf{y} = \textbf{x} \oplus (d, j)  
 \ \Leftrightarrow \    \textbf{y} _i=\begin{cases}
       \textbf{x}_i, & i \neq j \\
       \textbf{x}_i + d & i=j
    \end{cases}
\end{equation}
In simpler terms, a variable, $\textbf{x}_i$, should be eliminated if it fails the dimension-$i$ feasibility test along $\boldsymbol{\Delta}$ of length $d \geq h$, where the descent direction $\boldsymbol{\Delta}$ is given by unconstrained optimizers.
After eliminating $\textbf{x}_i$ at $\textbf{x}^{(k)}$, its value will always be the same as $\textbf{x}^{(k)}_i$.
Section~\ref{section_how_to_eliminate} will introduce an adaptive strategy to find the elimination tolerance $d$.

\begin{Example}
\label{dimension_feasibility_test}
    We continue with Example~\ref{example_nmbo} and first consider an elimination tolerance $d=10^{-5}$. 
    By Definition~\ref{def_dimension_test}, $\textbf{c}^{(k)}$ passes dimension-1 (with $\textbf{c}_1$) and dimension-2 (with $\textbf{c}_2$) feasibility test. This is because both $(5.999+10^{-5}, 1.499)$ and $(5.999, 1.499+10^{-5})$ are feasible. 
    If $d=10^{-1}$, $\textbf{c}^{(k)}$ will fail the dimension-1 feasibility test (with $\textbf{c}_1$) but still pass the dimension-2 feasibility test (with $\textbf{c}_2$). 
\end{Example}

\subsection{Theoretical Analysis of Variable Elimination}
\label{section_with_decoupled_obj}
The design metric in many real-time systems is decoupled, meaning the metric can be assessed for each task individually, where a separate set of variables is associated with each task, such as the problems in Sections~\ref{energy_opt_intro} and~\ref{control_opt_intro}. 
The decoupled metric permits the following useful property: 
\begin{definition} [Strict descent step]
\label{def_strict_descent_step}
    At $\textbf{x}$, an update step $\boldsymbol{\Delta} \in \mathbb{R}^N$  is a strict descent step if
        \begin{equation}
        \forall i \in [0, N-1], \ \mathcal{F}(\textbf{x} \oplus (\boldsymbol{\Delta}_i , i)) \leq \mathcal{F}(\textbf{x})
    \end{equation}
\end{definition}

\begin{observation}
\label{strict_descent_theorem}
    If $\mathcal{F}(\textbf{x})=\sum_{i=0}^{N-1} \mathcal{F}_i(\textbf{x}_i)$, and if it is differentiable at $\hat{\textbf{x}}$, then 
    there exists $\zeta>0$ such that 
    $\boldsymbol{\Delta}= -\zeta \boldsymbol{\nabla}  \mathcal{F}$ is often a strict descent step.
\end{observation}
\begin{justification}
This is based on Taylor expansion:
\begin{align*}
    \forall i, \exists \zeta>0,  \mathcal{F}(\hat{\textbf{x}} \oplus (\boldsymbol{\Delta}_i , i)) &
    \approx \mathcal{F}(\hat{\textbf{x}}) + 
    \boldsymbol{\Delta}_i \boldsymbol{\nabla}\mathcal{F}_i(\hat{\textbf{x}}_i)\\
    &= \mathcal{F}(\hat{\textbf{x}}) - \zeta (\boldsymbol{\nabla}\mathcal{F}_i(\hat{\textbf{x}}_i) )^2 \leq \mathcal{F}(\hat{\textbf{x}})
\end{align*}
\end{justification}

Next, we introduce extra symbols for the theorems below. At a feasible point $\textbf{x} \in \mathbb{R}^N$, $\boldsymbol{\Delta} \in \mathbb{R}^N$ denotes a descent direction provided by an unconstrained optimizer. Notice that $\boldsymbol{\Delta} \in \mathbb{R}^N$ may not be a feasible direction. 
$\boldsymbol{\mathcal{S}}$ denotes the set of dimension indexes $i$ that pass the dimension-$i$ feasibility test with length $d \geq h$ along $\boldsymbol{\Delta}$. 

We use $\bar{\boldsymbol{\mathcal{S}}}$ to denote the set of indices that fail the test. 
Within the context of a dimension-i feasibility test and an update step $\boldsymbol{\Delta}$, we also introduce a vector ${\textbf{v}}_{i}^{\mathcal{D}}\in \mathbb{R}^N$ where its \ith{j} element is given as follows:
\begin{equation}
     {\textbf{v}}_{i,j}^{\mathcal{D}} =\begin{cases}
        \boldsymbol{\Delta}_j & j =i\\
        0 & \text{otherwise}
    \end{cases}
\end{equation}
\begin{theorem}
    If $\boldsymbol{\Delta}$ is a strict descent step, and $|\boldsymbol{\boldsymbol{\mathcal{S}}}|=q>0$, then all the ${\textbf{v}}_i^{\mathcal{D}}$ are feasible descent vectors at $\textbf{x}$.
\end{theorem}
\begin{proof}
    Results of Definition~\ref{def_strict_descent_step} and Theorem~\ref{strict_descent_theorem}.
\end{proof}
\begin{theorem}
\label{ve_proof}
    If both the objective function $\mathcal{F}(\textbf{x})$ and  the constraints in the optimization problem~\eqref{general_F}
    are differentiable at $\textbf{x}^{(k)}$, and $|\boldsymbol{\mathcal{S}}|=q>0$, then there exists $\bar{d} < d$ such that $\textbf{x}^{(k)}+ \boldsymbol{\hat{\Delta}}$ is feasible with no worse objective function values:
    \begin{align}
           \boldsymbol{\hat{\Delta}} \in \{ &  \sum_i^q \zeta_i \textbf{v}_i^{\mathcal{D}} \ | \
        \|\boldsymbol{\hat{\Delta}}\|_1 \leq \bar{d},\ \zeta_i \geq 0 \} 
        \label{delta_range} \\
        & \mathcal{F}(\textbf{x}^{(k)}+ \boldsymbol{\hat{\Delta}}) \leq \mathcal{F}(\textbf{x}^{(k)})
        \label{update_decrease_obj}
    \end{align}
\end{theorem}
\begin{proof}


Proven by applying the Taylor expansion to the objective functions and constraints. 
\end{proof}

Although the objective functions and schedulability constraints in real-time systems are not differentiable everywhere in $\mathbb{R}^N$, VE is still useful because there are many differentiable points $\textbf{x}$.
Actually, if we randomly sample a feasible point $\textbf{x}\in \mathbb{R}^N$, $\textbf{x}$ is more likely to be a differentiable point than a non-differentiable point because there are only limited non-differentiable points but infinite differentiable points.
During the optimization iterations, if $\textbf{x}^{(k)}$ is a differentiable point, then Theorem~\ref{ve_proof} states that VE can find a descent and feasible update direction to improve $\textbf{x}^{(k)}$.

\subsection{Generalization to Non-differentiable Objective Function}
In cases that NMBO terminates at a non-differentiable point $\textbf{x}^{(k)}$, we can generalize the dimension feasibility test (Definition~\ref{def_dimension_test}) as follows:
\begin{definition} [Dimension feasibility descent test]
    A solution $\textbf{x}  \in \mathbb{R}^N$ for optimization problem \eqref{general_F} passes dimension-j feasibility descent test of length $\textit{d}$ along the direction $\boldsymbol\Delta$ if 
    $\textbf{x} \oplus (\frac{ \boldsymbol\Delta_j}{|\boldsymbol\Delta_j|}d, j )$ is feasible, where $j \in [0, N-1]$ and 
\begin{equation}
    \mathcal{F}(\textbf{x} \oplus (\frac{ \boldsymbol\Delta_j}{|\boldsymbol\Delta_j|}d, j )) \leq \mathcal{F}(\textbf{x})
\end{equation}
\end{definition}
It is useful when non-sustainable response time analysis is included in the objective function.

\subsection{Performing Variable Elimination}
\label{section_how_to_eliminate}
Theorem~\ref{ve_proof} suggests that $\textbf{x}^{(k)}$ can be improved by gradient-based optimizers if $\textbf{x}_{\bar{\boldsymbol{\mathcal{S}}}}$ are transformed into constant values at $\textbf{x}^{(k)}$.

The elimination tolerance $d$ is found by first trying small values (e.g., the numerical granularity $h$). If no variables are eliminated, we can keep increasing $d$ (for example, 1.5x each time) until we can eliminate at least one variable. We use the same elimination tolerance for all the variables in our experiments for simplicity, though individual elimination tolerance can also be set up.
\begin{Example}
\label{example_ve}
Let's continue with Example~\ref{example_nmbo} and see how to select $d$ and variables ($\textbf{c}_1$ or $\textbf{c}_2$) to eliminate. 
We can first try to increase $\textbf{c}_1$ and $\textbf{c}_2$ by $d=10^{-5}$ respectively but then find that no constraints are violated. 
Therefore, $d$ is kept being increased (1.5x each time) until $d=10^{-5}\times1.5^{12}>0.001$. 
In this case, $\textbf{c}^{(k)}$ fails the dimension-1 feasibility test (with $\textbf{c}_1$) while still passing the dimension-2 feasibility test (with $\textbf{c}_2$). 
Hence, $\textbf{c}_1$ will be eliminated first.
In future iterations, $\textbf{c}_1$ will be a constant 5.999, and $\textbf{c}_2$ will become the only variable left.
\end{Example}

\subsection{Termination Condition}
The algorithm may terminate at a stationary point found by NMBO.
However, in most cases, after NMBO terminates in each iteration, VE will find new variables to eliminate due to the schedulability constraints. Eventually, NORTH will terminate after all the variables are eliminated. 

\begin{theorem}
\label{theorem_terminate}
    The number of iterations in NORTH is no larger than $N$, the number of variables.
\end{theorem}
\begin{proof}
    This is because the values of elimination tolerance $d$ are selected to guarantee that at least one variable will be eliminated in each iteration.
\end{proof}

\begin{Example}
    In Example~\ref{example_ve}, after eliminating $\textbf{c}_1$ at $5.999$, $\textbf{c}_2$ becomes the only left variable. In the next iteration, NMBO will run and terminate at the schedulability boundary when $\textbf{c}_2=15.89$, where  $\tau_2$'s response time $r_2=39.89 < 40$.
    The algorithm does not make further progress (\eg terminate at $\textbf{c}_2=15.99$) because the relative error difference between the last two iterations in NMBO is already smaller than $10^{-5}$ when $\textbf{c}_2=15.89$.
    
After NMBO terminates, VE proceeds to find new variables for elimination. The elimination tolerance $d=10^{-5} \times 1.5^{12}$ from Example~\ref{example_ve} is first tried but cannot eliminate any new variables.
Therefore, we keep increasing $d$ until $d=10^{-5}\times1.5^{23}>0.11$, where $\textbf{c}_2+d$ violates the scheduleability constraint. Therefore, $\textbf{c}_2$ is eliminated next.
After that, there are no variables to optimize, and so NORTH will terminate at 5.999, 15.89, which is very close to the optimal solution (6, 16).
\end{Example}


\begin{theorem}
\label{theorem_north_feasible}
Solutions found by NORTH are always feasible.
\end{theorem}
\begin{proof}
This can be seen by considering the iteration process shown in \Figref{main_framework_fig}:
The initial solution is feasible; NMBO only accepts feasible updates, while VE does not change the values of variables.
Besides, all the constraints are kept and remain unchanged throughout the iterations. Therefore, the theorem is proved.
\end{proof}

\subsection{Discussion and Limitations of Variable Elimination}

VE is useful when there are active scheduling constraints. In real-time systems, the optimal solution is often near the boundary of the schedulable region, and VE is, therefore, often helpful.

Another way to view VE is that variables usually have different sensitivity to schedulability constraints. Therefore, we must differentiate these variables (by dimension feasibility test) and make distinct adjustments (by variable elimination and NMBO).
In this case, VE finds the variables far from violating the schedulability constraints and allows the unconstrained optimizers to optimize these variables further.

VE sacrifices the potential performance improvements associated with the eliminated variables. However, exploiting these potential performance improvements could be complicated and computationally expensive (\eg projected gradients, see more in Noccedal~\etal~\cite{Nocedal2006NumericalO2}). Besides, the potential performance improvements are likely to be limited because NORTH showed close-to-optimal performance
in experiments.


\section{Hybrid optimization}
\label{section_hybrid_optimization}
Similar to the problem~\eqref{general_F}, we consider real-time system design problems subject to black-box schedulability constraints in this section. The major difference is that priority assignments are incorporated into optimization variables. 
Although priority assignments have been widely studied, the black-box schedulability constraint makes them significantly more difficult, and there is not much work addressing similar problems. 

As modern computation systems become more complicated, an algorithm that can automatically explore and exploit an abstract problem and gradually improve itself would be highly useful in practice. To that end, such exploration must be guided by a metric. 
Therefore, we propose to utilize \textit{response time} to guide the exploration of priority assignments to optimize a given objective function while satisfying black-box schedulability constraints. Since response time can be treated as a continuous variable (because it can have floating-point values), we can first formulate an artificial problem to optimize the response time. After that, we propose a heuristic algorithm to adjust priority assignments iteratively based on the changes in response time. Next, we provide a more formal description of this algorithm, starting with the problem definition.

\subsection{Problem Description}
The hybrid optimization problem considered in this section is given as follows:
\begin{align}
    \min_{\textbf{x}, \textbf{P}} \ \ \ \ & \  \mathcal{H}(\textbf{x}, \textbf{r}) 
    \label{general_F_discre}\\
     \textit{subject to}\ \ & \  \textbf{r}=\boldsymbol{ \mathcal{G}}(\textbf{x}, \textbf{P})\label{eq_rta_hybrid_general} \\
    & \  \forall \tau_i \in \boldsymbol{\tau}, \textbf{r}_i \leq D_i \label{schedulability_discrete_general} 
\end{align}
where the optimization variables include $\textbf{x}$ and $\textbf{P}$:
$\textbf{x}$ are the continuous variables such as the execution times or periods of tasks (Assumption~\ref{assumption_continuous}), $\textbf{P}$ are priority assignments. Given $\textbf{x}$ and $\textbf{P}$, the response time $\textbf{r}$ is decided by a black-box response time analysis~\eqref{eq_rta_hybrid_general}. Therefore, $\textbf{r}$ is not treated as \textit{free} optimization variables in the problem~\eqref{general_F_discre}, even if the objective function depends on it. We left the bounding constraints~\eqref{general_inequality_constraint} for simplicity, but it can be handled similarly to the schedulability constraint.

The priority assignments $\textbf{P}$ for a task set $\boldsymbol{\tau}$ are described by an ordered sequence of all the tasks. Tasks that appear earlier in the sequence have higher priority. The response time $\boldsymbol{\mathcal{G}}(\textbf{x}, \textbf{P})$ is a black-box function that depends on $\textbf{x}$ and $\textbf{P}$. 
An example application is given in Section~\ref{control_opt_intro}.

\begin{Example}
A priority assignment vector $\textbf{P}=\{\tau_1, \tau_2, \tau_0\}$ denotes that $\tau_1$ has the highest priority, $\tau_0$ has the lowest priority.
\end{Example}

\subsection{Optimization Framework Overview}

We generalize NORTH to NORTH+ to solve the hybrid optimization problem above. 
The overview of NORTH+ is shown in \figref{fig_overview_co_design}, where we adopt an iterative framework to optimize the continuous variables and priority assignments.
When optimizing continuous variables, the priority assignments are fixed. Similarly, the continuous variables are treated as constant when optimizing the priority assignments.

Although NORTH can optimize the continuous variables, it cannot be directly applied to optimize priority assignments. Therefore, this section introduces a new algorithm motivated by the following observation:
\begin{observation}
\label{obs_monotonic_rta_priority}
    After increasing a task $\tau_i$'s priority, its worst-case response time $r_i$ usually will not become longer.
\end{observation}

The observation above establishes a monotonic relationship that usually holds between priority assignments and response time. For example, it is one of three conditions on the schedulability analysis that can be combined with Audsley's algorithm, an efficient priority assignment procedure~\cite{davis2016review}. 
Therefore, since response times $\textbf{r}$ are continuous variables, we can first fix the values of the priority assignment variables $\textbf{P}$ and obtain an update step $\boldsymbol{\Delta}\textbf{r}$ on the response times, then adjust priority assignments to make $\boldsymbol{\Delta}\textbf{r}$ happen.

\begin{figure}[t]
\centering
\includegraphics[width=0.36\textwidth]{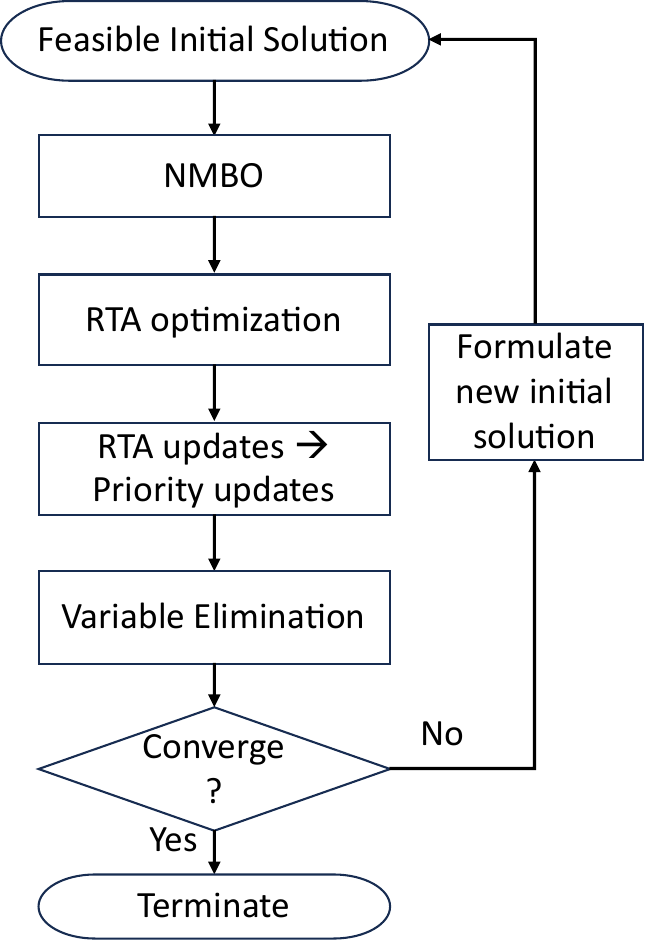}
\caption{Overall optimization framework NORTH+. In each iteration, the NMBO and VE components are shown in Figure~\ref{fig:main_framework_fig}, and the other two steps are introduced in this section.
The execution order of each step within an iteration is based on the following considerations: (1) variable elimination (VE) should be the last step because it reduces the variable space; (2) Numerical method-based optimization (NMBO) and priority assignments (RTA optimization and priority updates) optimize different types of variables, and so their execution order can be interchanged.}
\label{fig:fig_overview_co_design}
\end{figure}

Next, we introduce the major steps of NORTH+ in each iteration. 

\subsection{Step 1: Optimize the continuous variables with NORTH}
Since NORTH can optimize continuous variables, we first treat priority assignments $\textbf{P}$ as constants and optimize the continuous variables with NORTH:
\begin{align}
\textbf{x}^{(k+1), NORTH}  & \ \ =     \min_{\textbf{x}} \  
 \mathcal{H}(\textbf{x},  \textbf{r}) 
    \label{general_F_discre_north_step}\\
     \textit{subject to}\ \ & \  \textbf{r}=\boldsymbol{ \mathcal{G}}(\textbf{x},  \textbf{P})\Big|_{\substack{\textbf{P} = \textbf{P}^{(k)}}} \\
    & \  \forall \tau_i \in \boldsymbol{\tau}, \textbf{r}_i \leq D_i \label{schedulability_discrete_general_north_step} 
\end{align}
where $\boldsymbol{ \mathcal{G}}(\textbf{x}, \textbf{P})\Big|_{\substack{\textbf{P} = \textbf{P}^{(k)}}}$ denote that the priority assignment $\boldsymbol{P}$ is treated as a constant (specifically, its initial value at the start of each iteration) within the function $\mathcal{G}$. Since problem~\eqref{general_F_discre_north_step} does not involve discrete variables, we can utilize NORTH to solve it easily.

\subsection{Step 2: Optimize the Response Time Vector}
Since it is not easy to directly optimize priority assignments $\textbf{P}$ (because $\textbf{P}$ are discrete variables), we first solve an artificial unconstrained optimization problem:

\begin{align}
   \textbf{r}^{(k+1), art} & = \min_{\textbf{r}}  \ \ \ \ \  \mathcal{Z}(\textbf{r}  ,  \textbf{x})
    \conditionX \\
    & = \mathcal{H}(\textbf{r}  , \textbf{x})\conditionX- w\sum_{\tau_i \in \boldsymbol{\tau}} \textbf{Barrier}(D_i-\textbf{r}_i)
    \label{obj_with_barrier}\\
    \textit{subject to}\ \ & \ \ \ \ \ \ \ \ \ \ \ \ \ \ \ \ \ \ \ \ \ \ \ \ \ \ \ \ \ \ \ \   \emptyset
\end{align}
where $\mathcal{H}(\textbf{r}  ,  \textbf{x})\conditionX$ and $\mathcal{Z}(\textbf{r}  , \textbf{x})\conditionX$ denote that $\textbf{x}$ is treated as a constant variable in the functions. Different from problem~\eqref{general_F_discre}, $\textbf{P}$ is not considered in this step, and $\textbf{r}$ is the only optimization variable in this artificial problem~
\eqref{obj_with_barrier}. The major purpose of solving such an artificial problem is finding update directions for $\textbf{P}$.
We will add $\textbf{P}$ back in the next step.
$w>0$ controls the relative weight of the barrier function~\cite{Boyd2006ConvexO}:
\begin{equation}
    \textbf{Barrier}(x)=\begin{cases}
        \log(x), & x>0\\
        -\infty, & \text{otherwise}
    \end{cases}
\end{equation}
The transformed objective function above is safe because a good solver minimizes the objective function and avoids unschedulable situations.


The problem above can be solved by popular numerical optimizers such as LM (Equation~\eqref{lm_update}). We use gradient descent algorithm to solve it in this section for presentation simplicity. The gradient evaluation of the objective function~\eqref{obj_with_barrier} and the update step can be derived as follows:
\begin{equation}
    \boldsymbol{\nabla}_{\textbf{r}}  \mathcal{Z}
    =  \boldsymbol{\nabla}_{\textbf{r}} \mathcal{H}( \textbf{r} ,\textbf{x})\conditionX - w\sum_{\tau_i \in \boldsymbol{\tau}}\frac{1}{D_i-\textbf{r}_i}
    \label{eq_north_plus_gradient}
\end{equation}
\begin{equation}
    \boldsymbol{\Delta} \textbf{r} = - \eta \boldsymbol{\nabla}_{\textbf{r}}  \mathcal{Z}
\end{equation}
where $\eta$ controls the update step during iterations.
 
\subsection{Step 3: Optimize Priority Assignments }
Given an initial feasible priority assignment, this section presents a heuristic algorithm to adjust priority assignments to optimize $\mathcal{Z}(\textbf{r},  \textbf{x})\conditionX$ following $\textbf{r}^{(k+1), art}$. The hope is that we can change $\textbf{r}^{(k)}$ into $\textbf{r}^{(k+1), art}$ as much as possible by adjusting priority assignments $\textbf{P}^{(k)}$. 

The priority assignments are adjusted in an iterative way where we change one task's priority each time. 
At first, all the tasks are pushed into a set $\boldsymbol{\mathcal{S}}^{todo} = \boldsymbol{\tau}$, which contains the tasks waiting to adjust priority assignments. 
Within each iteration, we pick up one task from $\boldsymbol{\mathcal{S}}^{todo}$ and try to increase its priority until doing so cannot improve the objective function. 
The pseudocode of this idea is shown in Algorithm~\ref{alg_pa_opt}. 
The most important factor that decides the success of the algorithm is how to select tasks to increase its priority (the ``potential'' evaluation function in line \ref{alg2_pick_task_with_potential} in Algorithm~\ref{alg_pa_opt}), which is explained in Section~\ref{section_potential_hybrid_opt}.

\subsection{How to select tasks to increase its priority?}
\label{section_potential_hybrid_opt}
If the schedulability analysis is time-consuming, Algorithm~\ref{alg_pa_opt} may waste lots of time performing unnecessary schedulability analysis. Therefore, we propose two heuristics (Observation~\ref{obs_pa1} and \ref{obs_increase_pr_hash_table}) to decide how to select tasks that can improve the objective function.

We first introduce some notations. We use $\textbf{P}^{(k)}$ to denote the priority assignment at \ith{k} iteration, use $\textbf{P}^G$ to denote a priority assignment vector derived from $|\boldsymbol{\Delta} \textbf{r}|$ (Note the absolute value sign): tasks with bigger response time changes should have higher priority. 
Given a priority assignment vector $\textbf{P}$, $\textbf{P}(\tau_i)$ denotes the priority order of the task $\tau_i$. For example, the highest priority task's priority order is 1, and the lowest priority task's priority order is $N$ if there are $N$ tasks. 

\begin{Example}
\label{example_north_plus_intro}
    Let's consider the task set in Example~\ref{example_nmbo} except that we are optimizing task periods with the following objective function
    \begin{equation}
        \min_{\textbf{T}, \textbf{P}} T_1+T_2+r_1(\textbf{T},\textbf{P})+r_2(\textbf{T},\textbf{P})
    \end{equation}
    
    \noindent Furthermore, at the \ith{k} iteration, assume $\textbf{T}^{(k)}_1=10, \textbf{T}^{(k)}_2=6$. 
    When optimizing the period variables, we assume that the system follows the implicit deadline policy, which means that the deadline is always the same as the period variable. 
    The initial priority vector $\textbf{P}^{(k)}=\{1,2\}$, which means task 1 has a higher priority than task 2. 
    Assume the weight parameter $w$ in objective function~\eqref{obj_with_barrier} is $1$, then we have $\textbf{r}_1=4, \textbf{r}_2=5$, $\boldsymbol{\nabla}_{\textbf{r}}\mathcal{Z}^{(k)}=[1.25, 2]^T$,
    $\boldsymbol{\Delta} \textbf{r} = [-1.25, -2]^T$ (assume using a simple gradient descent algorithm and only perform one iteration),
    $\textbf{P}^G(\tau_1)=2$, $\textbf{P}^G(\tau_2)=1$.
\end{Example}

\begin{observation}
\label{obs_pa1}
Increasing a task $\tau_i$'s priority is more likely to improve the objective function if $\textbf{P}^G(\tau_i) < \textbf{P}^{(k)}(\tau_i) $
\end{observation}
\begin{justification}
    The basic idea of the observation is that a task with a bigger gradient should be assigned with a higher priority because it has more influence on the objective function. Therefore, a task with a smaller gradient should not be assigned with higher priority.
\end{justification}
\begin{Example}
    Continue with Example~\ref{example_north_plus_intro}. There is no need to increase $\tau_1$'s priority even if $\tau_1$ does not have the highest priority because its gradient is smaller than $\tau_2$'s.
\end{Example}
\begin{observation}
\label{obs_increase_pr_hash_table}
    If increasing the priority of a task $\tau_i$ from its current priority $j$ fails to improve the objective function, then it is highly likely that the same will hold even if other tasks' priority assignments change in future iterations.
\end{observation}
The second observation can be implemented via a hash map, which records the number of failed priority adjustments when $\tau_i$ is assigned the \ith{j} highest priority. If the failed adjustments exceed a threshold (1 in our experiments), we will not try to increase $\tau_i$'s priority to be the \ith{j} highest priority. Based on experiment results, maintaining the same hash map across iterations improves the overall running speed without sacrificing the overall performance.

\begin{theorem}
    Results found by NORTH+ always respect all the constraints in problem~\eqref{general_F_discre}.
\end{theorem}
\begin{proof}
    Similar to Theorem~\ref{theorem_north_feasible}, in each iteration, a new update in either the continuous variables or priority assignments is accepted only if it satisfies all the constraints and improves the objective function. Therefore, the final result always satisfies all the constraints.
\end{proof}

\subsection{Termination and convergence}
The termination conditions of NORTH+ are analyzed on the continuous variables and priority assignments separately. Iterations on continuous variables will terminate when all the variables are eliminated by Variable Elimination (VE). As for priority assignments, the iterations will terminate when no tasks are tried to increase its priorities based on Observation~\ref{obs_increase_pr_hash_table}.

The weight parameter $w$ in the objective function~\eqref{obj_with_barrier} can be set up following the classical interior point method. The initial $w$ parameter could have a relatively large value ($10^7$ in our experiments); during iterations, $w$ decreases by half after finishing one iteration every time. Eventually, $w$ will converge to 0 such that the transformed objective function~\eqref{obj_with_barrier} will converge to the actual objective function.

    
    


\subsection{Implementation Details}
The pseudocode for the adjusting priority assignments is shown in Alg~\ref{alg_pa_opt}. In line~\ref{alg_sort_tasks}, the tasks are sorted based on the absolute value of $\Delta \textbf{r}$, from the largest to the smallest. This is because prioritizing adjusting tasks with larger $|\Delta \textbf{r}|$ values is more likely to improve the objective function. In line~\ref{alg_line_change_priority_by_one}, if the $\Delta \textbf{r}_i$ is positive, we need to increase $\tau_i$'s priority; otherwise, we should decrease $\tau_i$'s priority. If one task $\tau_i$'s priority is increased by one, then the lowest priority task with higher priority than $\tau_i$ will have its priority decreased accordingly. For example, if the initial priority assignment vector from the highest to lowest is $\{\tau_0, \tau_1, \tau_2\}$, after increasing $\tau_2$'s priority by one, it becomes $\{\tau_0, \tau_2, \tau_1\}$.

\begin{algorithm}[ht!]
\SetAlgoLined
\SetKwInOut{a}{b}
\caption{Optimize priority assignments (One iteration)}
\label{alg_pa_opt}
\KwIn{ $\textbf{x}^{(k)}$, $\boldsymbol{P}^{(k)}$, expected response time change vector $\boldsymbol{\Delta \textbf{r}}$,  failing attempts record $\mathcal{M}(\tau_i, p) = \emptyset$, failing attempt threshold $\Theta^{\mathcal{M}}$}
\KwOut{$\boldsymbol{P}^{(k+1)}$}
\tcp{Add all tasks into a candidate list}
$\boldsymbol{\mathcal{S}}^{todo} = \boldsymbol{\tau}$ \\

\tcp{Sort the tasks in $\boldsymbol{\mathcal{S}}^{todo}$ based on the value of $|\boldsymbol{\Delta} \textbf{r}|$, largest to smallest}
$\boldsymbol{\mathcal{S}}^{todo} = \textbf{Sort}(\boldsymbol{\mathcal{S}}^{todo}, \boldsymbol{\Delta} \textbf{r})$ \label{alg2_pick_task_with_potential} \label{alg_sort_tasks} \\

\tcp{First consider the task with the most potential}
\For{$\tau_i$ in $\boldsymbol{\mathcal{S}}^{todo}$}
{       
    \tcp{If we have not failed much in trying to increase $\tau_i$'s priority before, we can try to increase $\tau_i$'s priority.}
    \While{$\mathcal{M}(\tau_i, \textbf{P}^{(k)}(\tau_i)) \leq \Theta^{\mathcal{M}}$} 
    { 
        $\boldsymbol{P} = \textbf{ChangeTaskPriorityByOne}(\boldsymbol{P}^{(k)}, \tau_i$, $\boldsymbol{\Delta}\textbf{r}_i$ ) \label{alg_line_change_priority_by_one} \\
        \eIf{ $\mathcal{H}(\xkth, \textbf{r}(\xkth, \boldsymbol{P})) < \mathcal{H}(\xkth, \textbf{r}(\xkth, \Pkth))$ }
        {
            $ \Pkth = \boldsymbol{P}$\\
        }{
            $\mathcal{M}(\tau_i,  \Pkth (\tau_i)) += 1$
        }
    }
    \textbf{RemoveTask}($\boldsymbol{\mathcal{S}}^{todo}$, $\tau_i$)
}
$\boldsymbol{P}^{(k+1)}=\boldsymbol{P}^{(k)}$
\end{algorithm}

In step 2, we need to optimize the response time $\textbf{r}$ to solve the artificial problem~\eqref{obj_with_barrier}. Typically, we can only adjust priority assignments to change response time from $\textbf{r}^{(k)}$ to $\textbf{r}^{(k+1), NORTH}$ up to a certain precision. Therefore, there is no need to solve the artificial optimization problem~\eqref{obj_with_barrier} with high precision, such as terminating iterations only when the relative difference between iterations becomes smaller than $10^{-5}$. In experiments, to improve run-time speed, we directly use the gradient vector as $\boldsymbol{\Delta} \textbf{r}$, which is essentially equivalent to solving the problem~\eqref{obj_with_barrier} with gradient descent and only performing one iteration.

\section{Run-time complexity estimation}
\label{run_time_complexity}
The run-time complexity of NORTH is analyzed as follows:
\begin{equation}
\label{eq_run_time_compexity}
    \text{Cost}_{NORTH}=N_{\text{Elimi}} \times (\text{Cost}_{\text{TR}} + \text{Cost}_{\text{Elimi}})
\end{equation}
where $N_{\text{Elimi}}$ denotes the number of elimination loops, $\text{Cost}_{\text{TR}}$ denotes the cost of the trust-region optimizer, and $\text{Cost}_{\text{Elimi}}$ denotes the cost of variable elimination. Let's use $N$ to denote the number of variables and analyze each of these terms.

Firstly, following Theorem~\ref{theorem_terminate}, we have
\begin{equation}
    N_{\text{Elimi}} \leq N
\end{equation}

Let's use $\text{Cost}_{\text{sched}}$ to denote the cost for schedulability analyses, $\text{N}_{\text{TR}}$ for the number of trust-region iterations, then:
\begin{equation}
    \text{Cost}_{\text{TR}} \leq N_{\text{TR}} \times (\text{Cost}_{\text{sched}} + O(N^3))
    \label{bound_tr}
\end{equation}
where $\text{N}_{\text{TR}}$ mostly is decided by the convergence rate of the gradient-based optimizer, which is often very fast, \eg super-linear or quadratic convergence rate~\cite{Nocedal2006NumericalO2}.
The $O(N^3)$ term above denotes the cost of matrix computation. 

As for the last term in Equation~\eqref{eq_run_time_compexity}:
\begin{equation}
    \text{Cost}_{\text{Elimi}} \leq N \times \text{Cost}_{\text{sched}}
\end{equation}

The run-time complexity of NORTH+ has an extra term regarding adjusting priority assignments:
\begin{equation}
    \text{Cost}_{NORTH+}=\text{Cost}_{NORTH} + N_{\text{Elimi}} \times O(N^2) \times \text{Cost}_{Sched}
\end{equation}
where the term $O(N^2)$ considers the maximum number of times each task can increase its priority.

\section{Application and generalizations}
\label{applications}
In this section, we discuss how to relax the two assumptions~\ref{assumption_feasible} and~\ref{assumption_continuous} and limitations.

\subsection{Finding feasible initial solutions}
\label{initial_solution_section}
The optimization variables in real-time systems usually have realistic meanings (such as period or execution time). 
In these cases, feasible initial solutions can be found by setting them the values that are more likely to be schedulable, \eg the shortest execution time/longest period. 
This method is optimal in finding initial solutions if the schedulability analyses are sustainable, such as the schedulability analysis based on the response time analysis~\eqref{rta_LL}. Please see more results in Section~\ref{section_exp_initial_solution}.
Alternatively, the feasible initial solution can be found by solving a phase-1 optimization problem~\cite{Boyd2006ConvexO}.
This method is applicable when the influence of violating schedulability constraints can be obtained quantitatively. 

\subsection{Optimizing Categorical Variables}
\label{optimize_discrete_variables}

Categorical variables are variables that represent data divided into categories or groups.
NORTH can optimize categorical variables if there is a rounding strategy to transform these variables into and back continuous variables. 
Variables that satisfy such requirements could be integer periods or run-time frequencies that can only be selected within a discrete set. In these cases, performing variable rounding is usually not difficult.
For example, if all the constraints are monotonic (in cases of schedulability constraints, that means they are sustainable), we can round floating-point variables into integers as follows: 
If rounding down the variable does not adversely affect system schedulability, we may do so to ensure schedulability; otherwise, we round it up.

Some schedulability analyses only take categorical variables as input, such as \cite{Nasri2019ResponseTimeAO} when optimizing periods (Although it is possible to use arbitrary float-point numbers as periods, the hyper-period would become very large and significantly increase the computation cost). In this case, continuous variables have to be rounded into categorical variables before performing the schedulability analysis.
This method is used in the NORTH+ experiment in Section~\ref{exp_north_control}.

\subsection{Potential to Optimizing Other Discrete Variables}
In addition to the categorical variables and priority assignments discussed previously, NORTH+ can also be employed to co-optimize continuous and some other discrete variables. This can be achieved through a similar iterative approach as described in Section~\ref{section_hybrid_optimization}, where continuous and discrete variables are optimized separately. However, this process requires an algorithm dedicated to optimizing discrete variables. A simple heuristic could suffice for this purpose. For instance, in the case of processor assignments, the algorithm might prioritize assigning processors to tasks with higher utilization.

\subsection{Limitations and potential solutions}
\label{limitation_sec}
NORTH relies on the gradient information from the objective function to utilize the gradient-based optimizers.
For example, if $\alpha_i=0$ for all the $i$ in the objective function~\eqref{ls_control}, then the gradient of \eqref{ls_control} with respect to the period variables would be either 0 or undefined. In this case, NORTH probably cannot make any progress. One potential solution to this issue is finding proxy variables or approximating the objective functions with those with meaningful gradient definitions with respect to the variables.

One limitation for NORTH+ is that it is hard to provide a theoretical guarantee on convergence because of the usage of heuristic algorithms to optimize the discrete variables. In general, such heuristic algorithms cannot guarantee to always improve the objective function, while the black-box schedulability constraints make such a guarantee even harder. 
In practical implementation, however, guaranteeing convergence can be easily achieved by terminating the iterations as soon as the new iteration cannot improve upon the previous iteration.

\section{Experiments}
NORTH and NORTH+ are implemented\footnote{
\href{https://github.com/zephyr06/Energy_Opt_NLP}{https://github.com/zephyr06/Energy$\_$Opt$\_$NLP}
} in C++, where the numerical optimization algorithms are provided by a popular library GTSAM~\cite{Dellaert12book_FactorGA} in the robotics area~\cite{Dellaert2017FactorGF, Kaess2008iSAMIS, Mukadam18ar_steap, Wang2020RobotCU}.
We perform the experiments on a desktop (Intel i7-11700 CPU, 16 GB Memory) with the following methods:
\begin{itemize}
    \item Zhao20, from \cite{Zhao2020AnOF}. It requires the schedulability analysis to be sustainable to be applicable.
    If not time-out, it finds the optimal solution.
    \item MIGP, mixed-integer geometric programming. The MIGP is solved by gpposy~\cite{gpposy} with BnB solver in YALMIP~\cite{Lofberg2004YALMIPA}.
    \item IPM, a famous constrained optimization method~\cite{Nocedal2006NumericalO2}. It is implemented with IFOPT and IPOPT~\cite{Wchter2006OnTI, ifopt}. 
    Numerical gradients are used in experiments to be applicable.
    \item NMBO, introduced in Section~\ref{section_nmbo}.
    \item NORTH, the optimization framework shown in \Figref{main_framework_fig}.
    \item SA, simulated annealing. The cooling rate is 0.99, the temperature is $10^5$, and the iteration limit is $10^6$.
\end{itemize}
LM~\eqref{lm_update} is the unconstrained optimizer in all the NORTH-related methods.
There is a time limit of 600 seconds for each method. 
We use the initial solution during performance evaluation if a baseline method cannot find a feasible solution within the time limit. 
The initial solutions in all the experiments are obtained following the heuristics in Section~\ref{initial_solution_section}.
Some methods are not shown in the figures because they are not applicable or require too much time to run. 


\begin{figure*}[ht!]
    \centering 
     \begin{subfigure}[t]{0.49\textwidth} 
        \centering
        \includegraphics[height=0.8\linewidth]{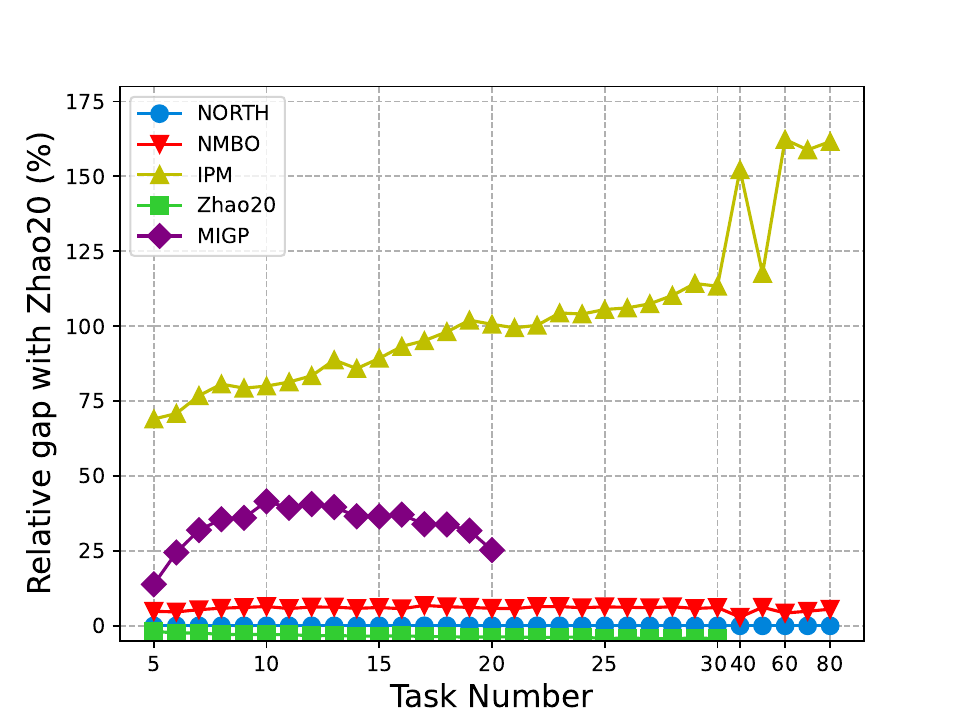}
        \caption{DVFS performance with RM Scheduling}
        \label{fig:compare_energy_ll}
    \end{subfigure}
    \begin{subfigure}[t]{0.49\textwidth}
        \centering
        \includegraphics[height=0.8\linewidth]{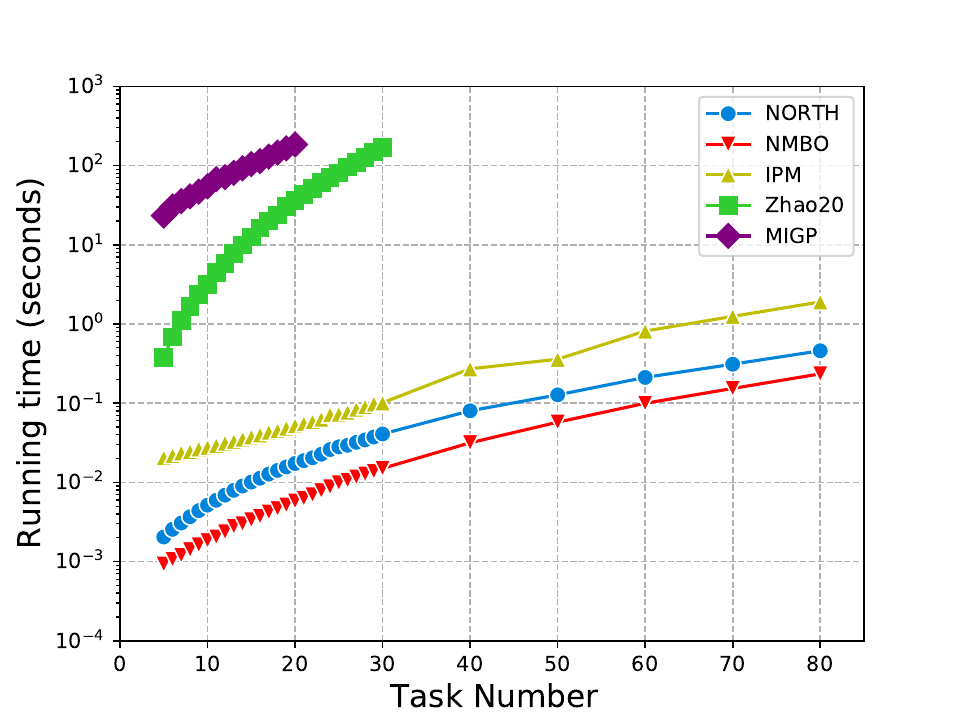}
            \caption{DVFS log run-time  with RM Scheduling}
    \label{fig:compare_energy_speed_ll}
    \end{subfigure}
\begin{subfigure}[t]{0.49\textwidth} 
        \centering
        \includegraphics[height=0.8\linewidth]{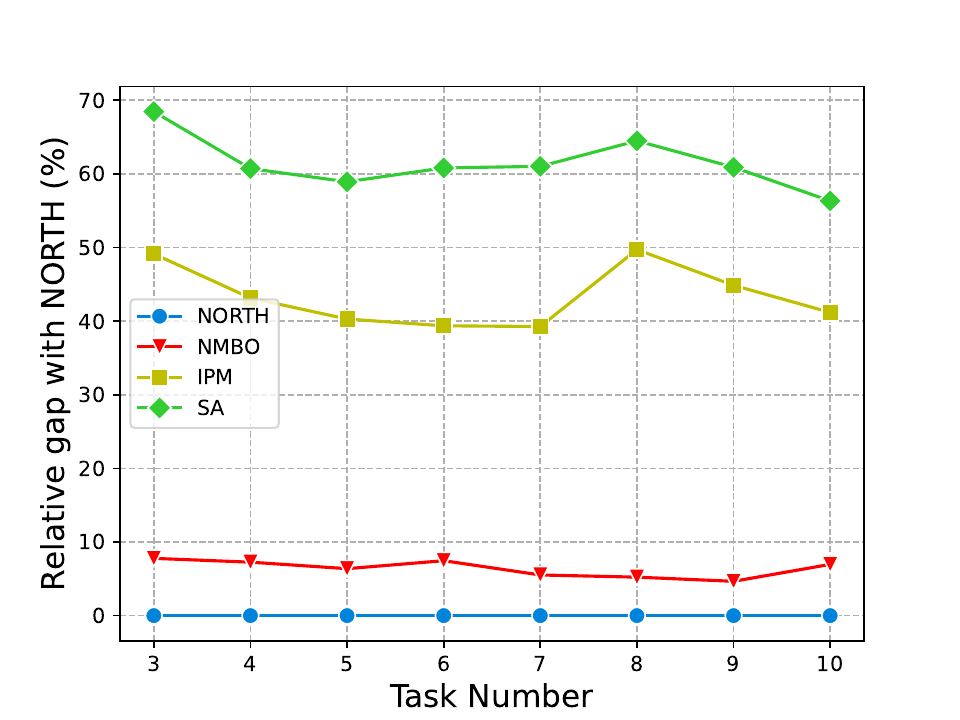}
        \caption{DVFS performance with DAG~\cite{Nasri2019ResponseTimeAO}}
        \label{fig:compare_energy_dag}
    \end{subfigure}
    \begin{subfigure}[t]{0.49\textwidth}
        \centering
        \includegraphics[height=0.8\linewidth]{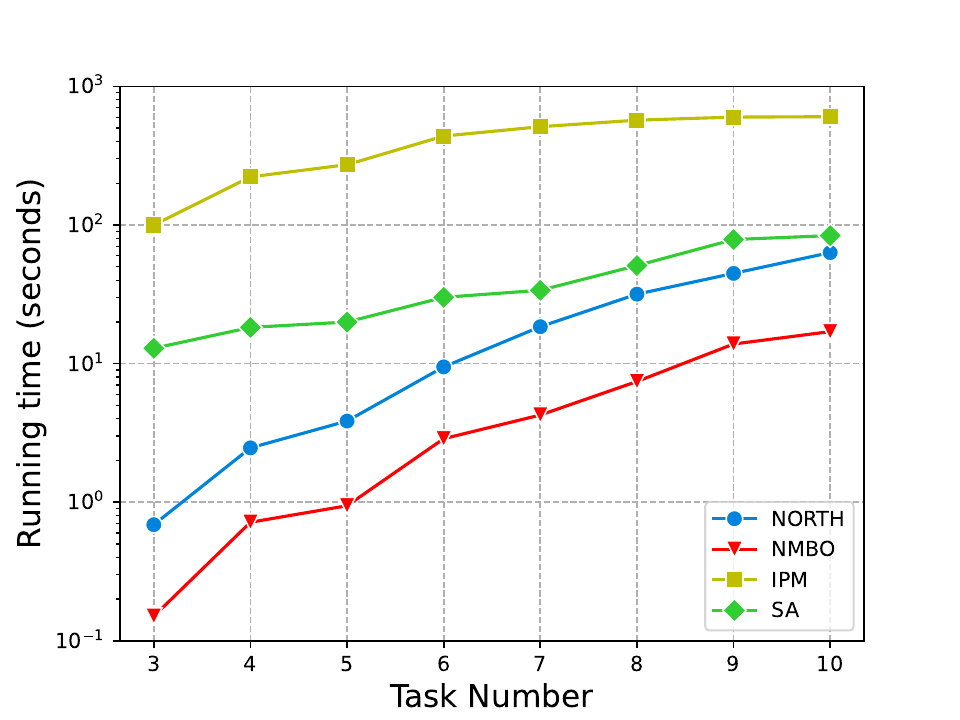}
    \caption{DVFS log run-time with DAG \cite{Nasri2019ResponseTimeAO}}
    \label{fig:compare_energy_speed_dag}
    \end{subfigure}

     \begin{subfigure}[t]{0.49\textwidth} 
        \centering
        \includegraphics[height=0.8\linewidth]{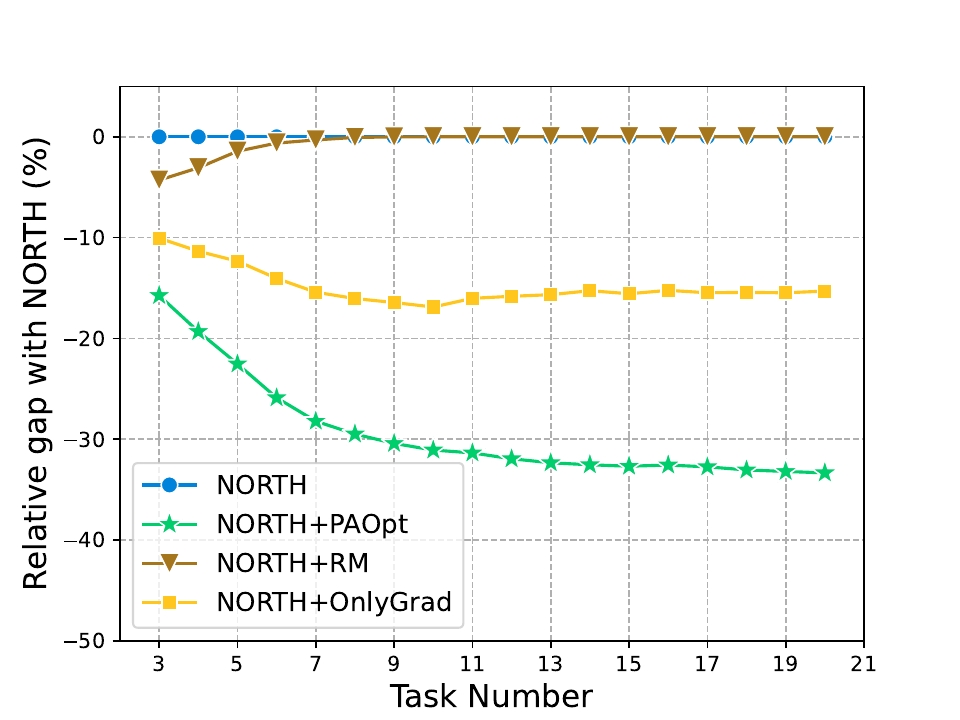}
        \caption{Hybrid control performance optimization}
        \label{fig:compare_control_hybrid_performance}
    \end{subfigure}
    \begin{subfigure}[t]{0.49\textwidth}
        \centering
        \includegraphics[height=0.8\linewidth]{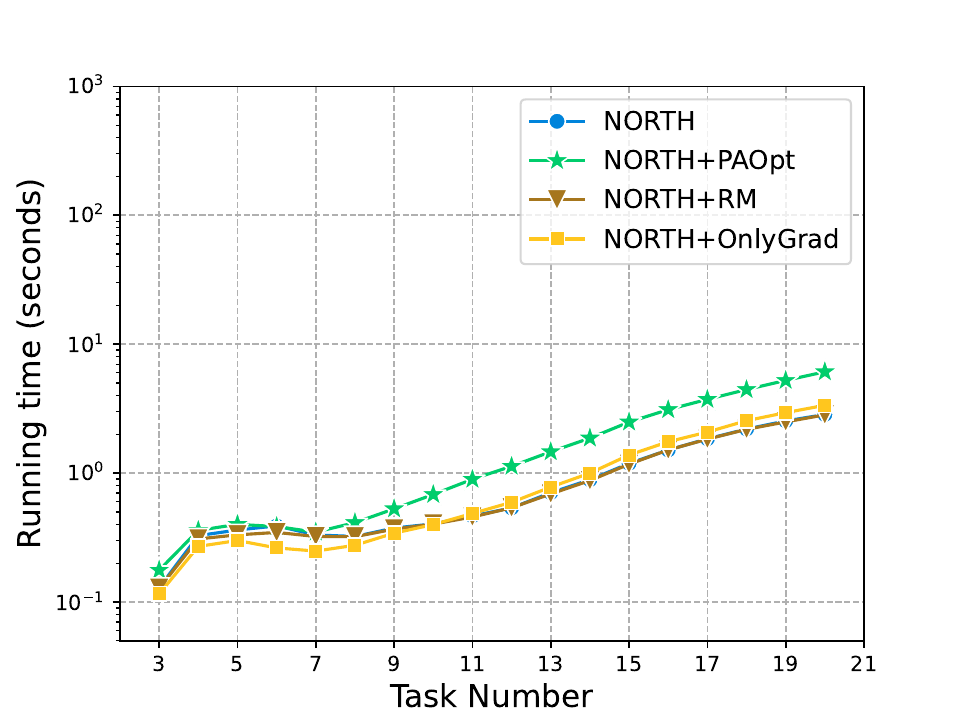}
            \caption{Hybrid control optimization log run-time}
    \label{fig:compare_control_hybrid_time}
    \end{subfigure}


    
    \caption{
   The performance and run-time speed of the experiments. 
    }
\end{figure*}

\subsection{Energy Optimization based on RM Scheduling}
In this experiment, we adopted the same settings as in Zhao~\etal~\cite{Zhao2020AnOF} for the convenience of comparison. The objective function is simplified as follows:
\begin{equation}
    E_i(\textbf{f}) =\frac{H}{T_i} ( \alpha_E \textbf{f}_i^{3}) \times ( \frac{c_i^{\text{org}}}{\textbf{f}_i})
\end{equation}
where $c_i^{\text{org}}$ denotes the WCETs for $\textbf{f}_i=1$.

We generate the task sets randomly as follows. The task set's utilization was randomly generated from $0.5$ to $0.9$.
For each task, its period was selected from a log-uniform distribution from $10^2$ to $10^5$; its utilization was generated using Uunifast~\cite{Davis2008EfficientES}; 
its run-time frequency $f_i$ is lower-bounded by half of its initial value $f_i^{(0)}$ (\ie $f_i \geq 0.5 f_i^{(0)}$);
its deadlines were the same as its period; its priorities were assigned based on the Rate Monotonic (RM) policy. The response time analysis~\eqref{rta_LL} and~\eqref{sched_model_ll} provide the schedulability constraints.

In \Figref{compare_energy_ll}, the relative gap between a baseline method and Zhao20 is defined as follows:
\begin{equation}
    \frac{E_{\text{baseline}}-E_{\text{Zhao20}}}{E_{\text{Zhao20}}} \times 100\%
    \label{experiment_criteria}
\end{equation}

\subsection{Energy Optimization for DAG Model}
\label{energy_opt_dag_exp_section}
In the second experiment, we consider the exact optimization problem~\eqref{ls_energy} where the scheduling algorithm and the schedulability analysis are provided by Nasri~\etal~\cite{Nasri2019ResponseTimeAO}. 
The schedulability analysis considers a task set whose dependency is modeled by a directed acyclic graph (DAG). The task set is executed on a non-preemptive multi-core computation platform.
The schedulability analysis is not sustainable with respect to either tasks' worst-case execution time or periods. 
Therefore, MIGP or Zhao20 cannot be applied to solve this problem.

The simulated task sets were generated randomly, where each task set contained several periodic DAG tasks. 
We generate the DAG structures following~\cite{He2021ResponseTB}, which randomly added an edge from one node to another with a probability (0.2 in our experiments). 
The number of nodes within each DAG follows a uniform distribution between 1 and 20. 
Nodes of the same DAG have the same period.
The period of each DAG was randomly selected according to a popular automotive benchmark~\cite{Kramer15benchmark} in the literature~\cite{Verucchi2020LatencyAwareGO, Bozhko2021MonteCR}. 
Considering the current schedulability analysis is time-consuming, the actual periods were randomly selected within a sub-set: $\{1, 2, 5, 10, 20, 50, 100\}$.  
The number of DAGs, $N$, in each task set ranges from 3 to 10. For each $N\in [3, 10]$, 105 task sets were generated, where 15 task sets were generated for each total utilization ranging from $0.1\times 4$ to $0.9 \times 4$. 
All the task sets were assumed to execute on a homogeneous 4-core computing platform. 
The utilization of each DAG within a task set and the execution time of each node are all generated with a modified Uunifast~\cite{Davis2008EfficientES} algorithm, which makes sure that each DAG/node's utilization is no larger than 100$\%$. 
During optimization, the relative error tolerance is $10^{-3}$, and the initial elimination tolerance is 10.

The execution time of all the nodes of all the DAGs are the optimization variables.
When $N=10$, there were 100 nodes on average to optimize, 200 nodes at most. 
\Figref{compare_energy_dag} and \Figref{compare_energy_speed_dag} show the results, where the relative performance gap is calculated against the initial energy consumption.

\subsection{Control Performance Optimization}
\label{exp_north_control}
Our third experiment considers the control performance optimization problem in Section~\ref{control_opt_intro}.
This problem is more challenging because the objective function contains non-differentiable response time functions; furthermore, both continuous variables (periods) and discrete variables (priority assignments) need to be optimized. 

The schedulability analysis is again given by Nasri~\etal~\cite{Nasri2019ResponseTimeAO}. 
The simulated task sets were generated similarly to the energy optimization experiment above, except for the following. The execution time of each node is generated randomly within the range [1, 100]. Within each task set, the initial period of all the nodes of all the DAG is set as the same value: five times the sum of all the nodes' execution time, rounded into a multiple of 1000.   
In the objective function~\eqref{ls_control}, the random parameters are generated following Zhao~\etal~\cite{Zhao2020AnOF} except that we add a small quadratic term: $\alpha_i$ was randomly generated in the range [1, $10^3$], $\beta_i$ was generated in the range [1, $10^4$], $\gamma_i$ is randomly generated in the range [-10, 10] ($\gamma_i$ is much smaller than $\alpha_i$ and $\beta_i$ following the cost function plotted in Mancuso~\etal~\cite{Mancuso2014OptimalPA}).

Since the schedulability analysis is sensitive to the hyper-period of the task sets, we limit the choice of tasks' period parameters into a discrete set $\{1,2,3,4,5,6,8\} \times \{100,1000,10000\}$ during optimization. The choices of period parameters are similar to the literature~\cite{Zeng2013AnEF, Kramer15benchmark} but provide more options.

The optimization variables are the periods of all the DAGs and the priorities of all the nodes of all the DAGs. The period of all the nodes within the same DAG is assumed to be the same. 
When $N=20$, there are $20$ period variables and 200 priority variables on average. 

To the best of our knowledge, no known work considers similar problems. Therefore, we add three more baseline methods to compare the performance of priority assignments. 
These methods follow the basic framework of NORTH+ in \Figref{fig_overview_co_design} except that the \textit{discrete optimization} step is performed based on different heuristics:
\begin{itemize}
    \item NORTH+RM: assign priorities to tasks based on RM.
    \item NORTH+OnlyGrad: assign priorities to tasks based on the gradient of objective functions with respect to the period variables. Tasks with bigger gradients have higher priority. 
    In other words, the $w$ parameter in equation~\eqref{eq_north_plus_gradient} is 0.
    \item NORTH+PAOpt, the priority optimization algorithm proposed in Section~\ref{section_hybrid_optimization}
\end{itemize}

The results are shown in \Figref{compare_control_hybrid_performance} and \Figref{compare_control_hybrid_time}.

\subsection{Result Analysis and Discussions}
\subsubsection{NORTH analysis}
The first two experiments showed that NORTH can achieve excellent performance while maintaining fast runtime speed. Compared with the state-of-the-art methods Zhao20~\cite{Zhao2020AnOF} (it finds global optimal solutions if not time-out), NORTH maintains similar performance (around 1$\sim$3$\%$) while running $10^2\sim10^5$ times faster in example applications. 
Furthermore, in cases of tight time budget, NORTH could achieve even better performance in some cases.
NORTH also outperforms the classical numerical optimization algorithm IPM in terms of both performance and speed because IPM cannot directly handle the non-differentiable schedulability constraint~\eqref{schedulability_analysis_true_false}, even with numerical gradients.
Besides, NORTH shows more stable performance improvements under different situations than SA and IPM (SA adopts a stochastic searching strategy; numerical gradient at non-differentiable points in IPM is not predictable).
Finally, our experiments also show that NORTH supports optimizing large systems with fast speed.

\subsubsection{NORTH+ analysis} 
\Figref{compare_control_hybrid_performance} compares the performance between NORTH and NORTH+.
Built upon NORTH, NORTH+ supports optimizing continuous and discrete variables and further improves performance. Compared with simple heuristics such as RM and OnlyGrad, the proposed priority assignment algorithm achieved the best performance because it properly balanced both the objective function and the system's schedulability requirements. 
\Figref{compare_control_hybrid_time} shows the run-time speed of the baseline methods. NORTH+ runs slower because it has to perform the extra priority assignments. However, all the algorithms have similar and good scalability as the task set scales bigger.

\subsection{Obtaining feasible initial solutions} 
\label{section_exp_initial_solution}
In our experiments, the initial solution was derived using a simple heuristic: selecting tasks with the shortest execution time and longest period. While optimal for sustainable schedulability analysis, this approach may fail in non-sustainable cases.

To evaluate its effectiveness, we conducted an additional experiment (Based on section~\ref{energy_opt_dag_exp_section}), considering three DAGs per task set, each with an average of 10 nodes. A 100-second limit was imposed for verifying feasible solutions, and task sets were discarded if this limit was exceeded. Overall, 3343 out of 4500 random task sets were schedulable, with the heuristic identifying feasible solutions in 96\% of these cases. Detailed statistics are presented in Table~\ref{initial_feasible_table}.


\begin{table}[ht]
\centering
\caption{
Proportion of schedulable random DAG task sets using shortest execution time (Schedulability analysis from \cite{Nasri2019ResponseTimeAO}, across per-core average utilization).}
\begin{tabular}{@{}cccccccccc@{}}
\toprule
\begin{tabular}[c]{@{}c@{}}Utilization (\%)\end{tabular} & 10  & 20  & 30  & 40  & 50 & 60 & 70 & 80 & 90 \\ \midrule
\begin{tabular}[c]{@{}c@{}}Proportion (\%)\end{tabular}   & 100 & 100 & 100 & 100 & 97 & 93 & 88 & 72 & 47 \\ \bottomrule
\end{tabular}%
\label{initial_feasible_table}
\end{table}


\section{Conclusion}
In this paper, we introduce NORTH, a general and scalable framework for real-time system optimization based on numerical optimization methods. 
NORTH is general because it works with arbitrary black-box schedulability analysis, drawing inspiration from classical numerical optimization algorithms. 
However, NORTH distinguishes itself from them by how to identify and manage active constraints while satisfying the schedulability constraints.
Furthermore, novel heuristic algorithms are also proposed to perform collaborative optimization of continuous and discrete variables (\eg categorical variables and priority assignments) iteratively guided by numerical algorithms.
The framework's effectiveness is demonstrated through two example applications: energy consumption minimization and control system performance optimization. 
Extensive experiments suggest that the framework can achieve similar solution quality as state-of-the-art methods while running $10^2$ to $10^5$ times faster and being more broadly applicable.


\bibliographystyle{ieeetr}
\bibliography{scheduling} 

\end{document}